\documentclass[sigconf,edbt]{acmart-edbt2022}

\def\BibTeX{{\rm B\kern-.05em{\sc i\kern-.025em b}\kern-.08em
    T\kern-.1667em\lower.7ex\hbox{E}\kern-.125emX}}

\usepackage{booktabs} 

\usepackage{setspace}
\usepackage{placeins}
\usepackage{afterpage}
\usepackage[utf8]{inputenc}
\usepackage{graphicx}
\usepackage{balance}
\usepackage{amsthm}
\usepackage{color}
\usepackage{caption}
\usepackage{amsmath}
\usepackage{tabu}
\usepackage{array}
\usepackage{booktabs}
\usepackage{threeparttable}
\usepackage{relsize}
\usepackage{physics}

\usepackage{algorithm}
\usepackage[noend]{algpseudocode}

\newlength{\maxwidth}
\newcommand{\algalign}[2]
{\makebox[\maxwidth][r]{$#1{}$}${}#2$}

\usepackage{float}
\usepackage{stfloats}
\usepackage{lipsum}
\usepackage[english]{babel}

\usepackage{comment}

\usepackage[T1]{fontenc}
\usepackage{mwe}    
\usepackage{subfig}

\let\oldnl\nl
\newcommand{\nonl}{\renewcommand{\nl}{\let\nl\oldnl}}

\newtheorem{problem}{Problem}

\theoremstyle{definition}
\newtheorem{defn}{Definition}

\def\Equal{\texttt{=}}

\setcopyright{rightsretained}

\acmDOI{}

\acmISBN{978-3-89318-086-8}

\acmConference[EDBT 2022]{25th International Conference on Extending Database Technology (EDBT)}{29th March-1st April, 2022}{Edinburgh, UK} 
\acmYear{2022}

\begin{document}

\title{Differentially-Private Publication of Origin-Destination Matrices with Intermediate Stops}

\author{Sina Shaham}
\affiliation{%
  \institution{University of Southern California}
  \state{Los Angeles, California} 
  \country{USA}
}
\email{sshaham@usc.edu}

\author{Gabriel Ghinita}
\affiliation{%
  \institution{University of Massachusetts}
  \state{Boston, Massachusetts} 
  \country{USA}
}
\email{gghinita@cs.umb.edu}

\author{Cyrus Shahabi}
\affiliation{%
  \institution{University of Southern California}
 \state{Los Angeles, California} 
  \country{USA}
}
\email{shahabi@usc.edu}

\renewcommand{\shortauthors}{}

\begin{abstract}
Conventional origin-destination (OD) matrices record the count of trips between pairs of start and end locations, and have been extensively used in transportation, traffic planning, etc. More recently, due to use case scenarios such as COVID-19 pandemic spread modeling, it is increasingly important to also record intermediate points along an individual's path, rather than only the trip start and end points. This can be achieved by using a multi-dimensional {\em frequency matrix} over a data space partitioning at the desired level of granularity. However, serious privacy constraints occur when releasing OD matrix data, and especially when adding multiple intermediate points, which makes individual trajectories more distinguishable to an attacker. To address this threat, we propose a technique for privacy-preserving publication of multi-dimensional OD matrices that achieves differential privacy (DP), the de-facto standard in private data release. We propose a family of approaches that factor in important data properties such as data density and homogeneity in order to build OD matrices that provide provable protection guarantees while preserving query accuracy. Extensive experiments on real and synthetic datasets show that the proposed approaches clearly outperform existing state-of-the-art.
\end{abstract}

%
%



\maketitle

\section{Introduction}\label{Sec: Introduction}

Origin-destination (OD) matrices have been extensively used to characterize the demand for transportation between pairs of start and end trip points. Using OD matrices, one can provision appropriate capacity for a transportation infrastructure, by determining what is the demand (or trip {\em frequency}) for each source-destination pair. However, novel applications require more level of detail, for which conventional OD matrices are insufficient, due to the fact that they have a 2D structure, and intermediate points along a trajectory cannot be captured. Consider, for instance, the study of COVID-19 spread patterns in the ongoing pandemic, where an analyst needs to determine not only the end points of a trajectory, but also the intermediate points that a certain individual has visited, and where possible exposure to the virus occurred. In this case, it is necessary to record several distinct points across a trajectory, which leads to an increase in the dimensionality of OD matrices. We denote such enhanced data structures as OD matrices {\em with intermediate stops}.

While such detailed OD matrices capture additional information, they also pose a more serious privacy threat for the individuals included in the data, since the finer level of granularity of trajectory representation allows an adversary to pinpoint a user with better accuracy. For instance, there may be a large number of users that travel between a suburban neighborhood and the city center. However, when intermediate stops are also included, e.g., a specific type of store that sells ethnic products, a gym specializing on a certain type of yoga, and a fertility clinic, there are far fewer individuals who follow such a path (and sometimes, perhaps just one individual), which may lead to serious privacy breaches related to that individual's gender, race and lifestyle details. It is thus essential to protect the privacy of individuals whose trajectories are aggregated to build detailed OD matrices, and {\em differential privacy (DP)}~\cite{dwork2006calibrating} is the model of choice to achieve an appropriate level of protection. 

Specifically, DP bounds the ability of an adversary such that s/he cannot determine with significant probability whether the trajectory data of a target individual is present in the released OD matrix or not. The OD matrix with intermediate stop points is equivalent to a {\em multi-dimensional frequency matrix}, in which an element represents the number of individuals who took a trip that includes that specific sequence of start, intermediate and end points. According to DP, carefully calibrated noise is added to each count to bound the identification probability of any single individual. 

Several approaches tackled the problem of protecting frequency matrices for location data, but they do have serious limitations. For instance, solutions for DP-compliant location data histograms~\cite{AG, cormode2012differentially, zhang2014towards,zhang2016privtree} build data-independent structures that do not adapt well to data density, and assume a fixed dimensionality of the indexing structure, typically 2D only. As we show in Section~\ref{Experimental Evaluation}, they do not handle well skewed datasets, which are the most typical ones in the case of geospatial data. Another category of approaches attempts to capture trajectories using prefix trees or $n$-grams~\cite{acs2012differentially, chen2012differentially}, but those approaches transform cells in the data domain into a sequence of abstract string labels, and lose the proximity semantics that are so important when querying location-based data.

We propose a novel technique for sanitization of OD matrices with intermediate stops such that location proximity semantics are preserved, and at the same time the characteristics of the data are carefully factored in to boost query accuracy. Specifically, we build custom data structures that tune important characteristics like index fan-out and split points to account for data properties. This way, we are able to achieve superior accuracy while at the same time enforcing the strong protection guarantees of DP.

Our specific contributions are:
\begin{itemize}
    \item We identify important properties of indexing data structures that have a high impact on query accuracy when representing location frequency matrices; 
    \item We design customized approaches that are guided by intrinsic data properties and automatically tune structure parameters, such as fan-out, split points and index height;
    \item We perform a detailed analysis of the obtained data structures that allows us to allocate differential privacy budget in a manner that is conducive to preserving as much data accuracy as possible under a given privacy constraint;
    \item We perform an extensive experimental evaluation on both real and synthetic datasets which shows that our proposed techniques adapt well to data characteristics and outperform existing state-of-the-art in terms of query accuracy.
\end{itemize}

Section~\ref{Sec: system model} introduces necessary background concepts and definitions. Section~\ref{Sec: Data Independent Approaches} outlines data-independent techniques, followed by data-adaptive approaches in Section~\ref{Sec: Data Dependent Approaches}. 
Section~\ref{Sec: literature reivew} surveys related work.
We present experimental evaluation results in Section~\ref{Experimental Evaluation}, followed by conclusions in Section~\ref{Sec: Conclusion}.

\section{Background and Definitions}\label{Sec: system model}

We assume the two-party system model shown in Fig.~\ref{Fig: system model}: a trusted data curator/owner collects the frequency matrix directly from individuals and sanitizes the data. Untrusted data analysts are interested in querying the private frequency matrix.

Let $F_1\times F_2 \times ...\times F_d$ be a $d$-dimensional array representing a frequency matrix $F$. Each entry $f_i\in F$ is a number denoting a {\em frequency} or {\em count}. For example, a two-dimensional frequency matrix can model a map with each entry indicating the number of individuals located in a particular area. The frequency matrix corresponds to a $d$-dimensional finite space hyper-rectangle, or $d$-orthotope. According to the differential privacy model, a protection mechanism adds to each matrix element noise from a carefully selected random distribution to prevent an adversary from learning with significant probability whether a particular individual's data was used or not when creating the matrix.
\begin{figure}[t]
\raggedright
\includegraphics[scale=.4]{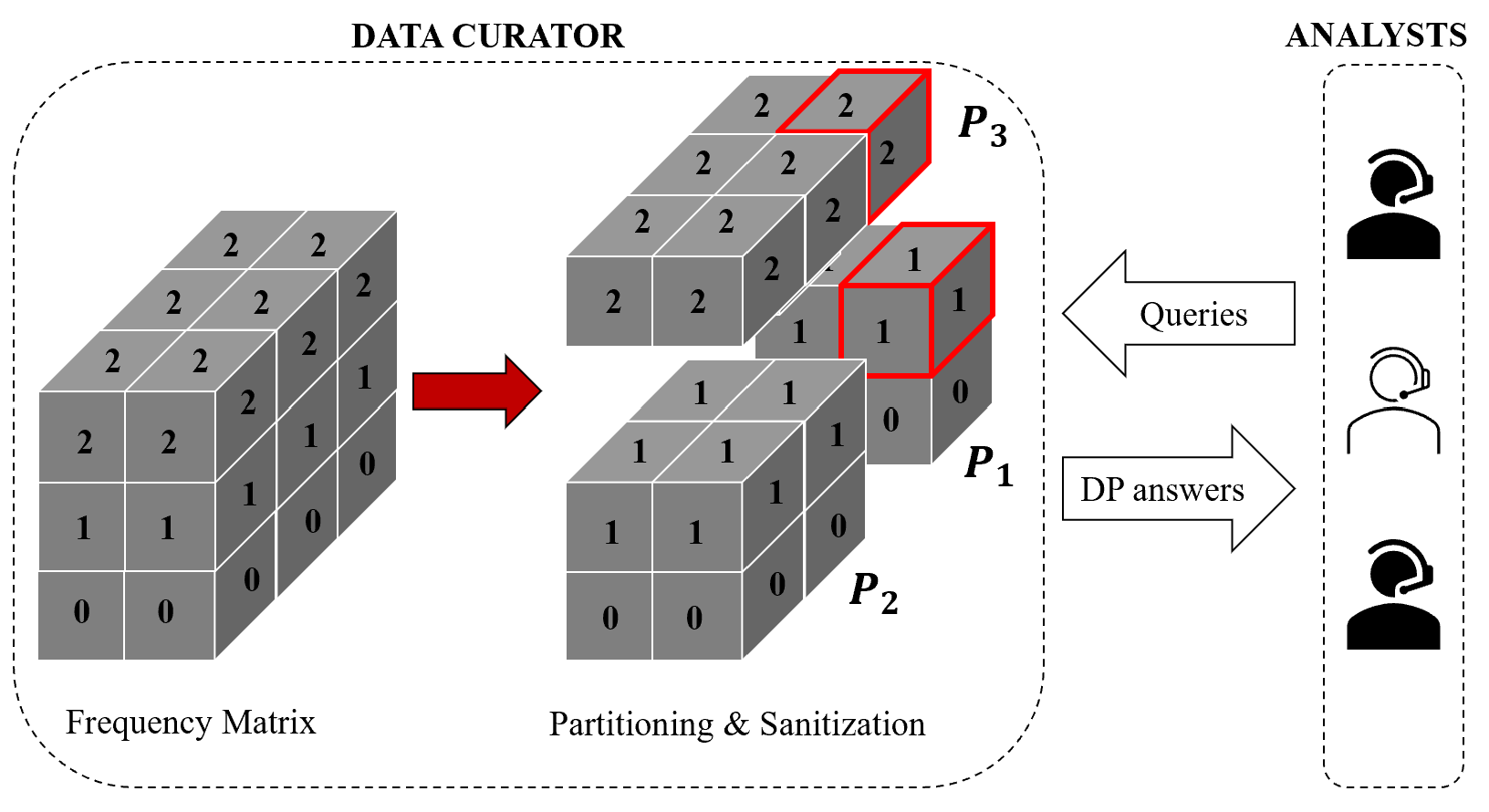}
\centering
\caption{System model for private frequency matrices.}
\label{Fig: system model}
\end{figure}

\subsection{Differential Privacy}
Differential privacy (DP)~\cite{dwork2006calibrating} is a popular privacy model which provides strong protection guarantees. It presents an aggregate query interface (i.e., count queries) and ensures that the presence or absence of an individual in the data does not significantly change the results of a query. Consider two frequency matrices $F$ and $F'$ that differ in a single record $t$, i.e., $F'=F\bigcup \{t\}$ or $F'=F \diagdown \{ t \}$. $F$ and $F'$ are commonly referred to as {\em neighboring} or {\em sibling} datasets. 

\begin{defn}[$\epsilon$-Differential Privacy] 
    A randomized mechanism $\mathcal{A}$ provides $\epsilon$-DP if for any pair of neighboring frequency matrices $F$ and $F'$, and any output value $S\in Range(\mathcal{A})$,
    \begin{equation}
        \dfrac{Pr(F=S)}{Pr(F'=S)} \leq e^\epsilon 
    \end{equation}
\end{defn}
Parameter $\epsilon$ is the {\em privacy budget}: lower values result in stricter privacy, but also require addition of noise with larger magnitude, decreasing query accuracy. 
The {\em sequential composability} property of DP states that running in succession multiple mechanisms that each satisfy DP with privacy budgets $\epsilon_1, \epsilon_2,...,\epsilon_n$ respectively, is equivalent to the execution of a single mechanism with $\epsilon = \sum_{i=1}^n \epsilon_i$.

An essential concept of DP is the {\em sensitivity} of queries, which measures the maximal difference that can be achieved by the addition or removal of a single individual's record in the database. 

\begin{defn}[$L_1$-Sensitivity]
Given two sibling datasets $F$, $F'$ and a set of real-valued functions $\mathcal{G}=\{ g_1,g_2,...,g_m \}$, the $L_1$-sensitivity of $F$ is measured as 
$$s\Equal \underset{\forall F,F'}{max}\sum_{i\Equal1}^m |g_i(F)- g_i(F')|$$
\end{defn}

The Laplace mechanism is a widely used technique to achieve $\epsilon$-DP. It adds to the output of a query function $g$ noise drawn from a Laplace distribution with scale $b$, where $b$ depends on two factors: sensitivity and privacy budget.

\begin{equation}
    Pr(x|b) = \dfrac{1}{2b}e^{|x|/b}\; \text{where }\; b=\dfrac{s}{\epsilon} 
\end{equation}
In the rest of the paper, we denote Laplace noise by $Lap( \dfrac {s } {\epsilon} )$. In the case of query functions that are modeled through a partitioning of the dataspace (e.g., a set of non-overlapping histogram bins), sensitivity is equal to $1$, since a record can fall in exactly one partition.


\newcommand{\rvec}{\mathrm {\mathbf {r}}} 
\begingroup
\begin{table}
\caption {Summary of notations.} 
\centering
\begin{tabular}{>{\arraybackslash}m{2cm} >{\arraybackslash}m{5.8cm} }
\hline\hline
  Symbol  & Description\\    \hline
  $F$ & Frequency matrix \\
  $F_i$ & Dimension cardinality\\
  $N$& Total count of $F$\\
  $ \overline{N} $ &  Sanitized total count of $F$\\
  $m$ & Partitioning constant\\
  $s$& Sensitivity\\
  $\epsilon_{\text{tot}}$ & Total privacy budget\\
  $\epsilon_{\text{prt}}$ & Partitioning budget\\ 
  $\epsilon_{\text{data}}$ & Data perturbation budget\\
  $H(F)$ & Entropy of $F$\\
  $ Lap(s/\epsilon)$ & Laplace noise with sensitivity $s$ and budget $\epsilon$\\
\hline\hline
\end{tabular}
\label{tab:table1}
\end{table}
\endgroup

\subsection{Problem Statement}

Starting with an input frequency matrix, we create a set of non-overlapping partitions of the matrix and then publish a set of noisy counts for each of these partitions, according to the Laplace mechanism. The {\em sanitized}, DP-compliant frequency matrix consists of the {\em boundaries} of all partitions and their {\em noisy counts}. Since partitions are non-overlapping, we keep sensitivity low (i.e., 1). We refer to each input cell in the original frequency matrix as an {\em entry}, hence a {\em partition} is a group of matrix entries.
Analysts (i.e., users of the sanitized matrix) ask multi-dimensional {\em range queries}. 
\begin{defn} {\em (Range Query)}
A range query on the frequency matrix $F$ is a $d$-orthotope with dimensions denoted as $d_1\times d_2 \times ...\times d_n$, where $d_i$ represents a continuous interval in dimension $i$. 
\end{defn}
For example, consider the $3\times2 \times 3$ frequency matrix shown in Fig.~\ref{Fig: system model}. 
The generation of partitions is referred to as {\em partitioning} and the addition of noise to total sums is referred to as {\em sanitization}. The example shows a sample partitioning of the matrix generating three partitions $P_1,\, P_2$ and $P_3$ with total counts of $2,\, 4$ and $12$, respectively. In a simplified setup, the sanitization follows by adding Laplace noise to the partitions' total count and answering queries based on the resulting private frequency matrix. Moreover, a {\em uniformity assumption}~\cite{cormode2012differentially} is made within each partition to answer queries with varying shapes and sizes. For example, if the sanitized counts are $2+n_1$, $4+n_2$, and $12 +n_3$, where $n_i$ denotes Laplace noise added for sanitization, and an analyst asks a query including two cells whose borders are shown in bold red color, the answer is $\dfrac{12 +n_3}{6}+\dfrac{2+n_1}{4}$.\\

Suppose that the total count of a partition entailing $q$ entries is $p$, and its noisy count is denoted by $\overline{p}$. One can see that there are two sources of error while answering a query. The first type of error is referred to as {\em noise error}, which is due to Laplace noise added to the partition counts. The second source of error is referred to as {\em uniformity error}. The uniformity error arises as the assumption of uniformity is made within each partition so that the noisy count of a cell in the partition can be calculated as $\overline{p}/q$.

To evaluate accuracy of query results, we use the {\em mean relative error (MRE)}.
For a query $q$ with the true count $p$ and noisy count $\overline{p}$, MRE is calculated as 
\begin{equation}\label{Equation: MRE}
   MRE(q) =  \dfrac{|p - \overline{p}|}{p}\times 100
\end{equation}

\begin{problem}\label{problem statement}
Given a frequency matrix $F$, generate an $\epsilon$-differentially private version of $F$ such that the expected value of relative error (MRE) is minimized. 
\end{problem}

In the design of methods for the publication of private frequency matrices, we make extensive use of {\em entropy} to understand the amount of information contained in the frequency matrix and the effect that partitioning has on information loss. 
\begin{defn} {\em (Entropy)}
Given a frequency matrix $F$ and a set of partitions $\mathcal{P} = \{P_1,\,P_2,...,\, P_n\}$ with the total counts $p_1,\,p_2,...,\, p_n$, the entropy of $F$ is defined as:
\begin{equation}
    H(F|\mathcal{P}) = -\sum_{i=1}^n \dfrac{p_i}{\sum_{j=1}^n p_j} \log_2{\dfrac{p_i}{\sum_{j=1}^n p_j}}
\end{equation}
\end{defn}
Table~\ref{tab:table1} summarizes notations used throughout the paper.

\subsection{Trajectory Modeling with OD Matrices}

Conventional OD matrices allow analysts to determine how many individuals traveled between pairs of locations, e.g., between the central business district (CBD) and a suburb. Increasing availability of mobile data and their use in complex planning problems makes it important to expand the expressiveness of OD matrices, by allowing one to include intermediate stops, which essentially amounts to supporting queries on trajectories. Furthermore, conventional OD matrices tend to use abstract representations of locations, where the spatial information may be lost, e.g., by tabulating counts of individuals traveling between pairs of zipcodes. Proximity of zipcodes may be lost in the process, and if one wishes to change the representation granularity, or perform range-based queries (e.g., find how many users traveled from a $1 km$ circle centered at point $A$ to a $1 km$ circle centered at $B$), such functionality is not possible.

Our proposed multi-dimensional histograms produce a hierarchical partitioning of the data domain that preserves locality and proximity information. It allows flexible queries, and captures intermediate points along a trajectory, as shown in Figure~\ref{fig:OD}. 
Assume a trajectory representation where one wishes to capture daily activities across several time frames, e.g., {\em morning->noon->evening}. Trajectory $T_1$ corresponds to a person who lives in a suburb, works in CBD and goes to see a play in the evening. This can be captured using a multi-dimensional histogram where the first pair of spatial coordinates corresponds to the morning location (suburb), followed by another pair in the CBD, and finally the evening in the theater district. Each of the time frames can be partitioned independently, resulting in the structure on the right half of Figure~\ref{fig:OD} (due to space constraints, we do not represent the evening time frame). Each trajectory corresponds to a single entry in this multi-dimensional matrix, according to each location at each time frame.

\begin{figure}[t]
	\includegraphics[width=\columnwidth]{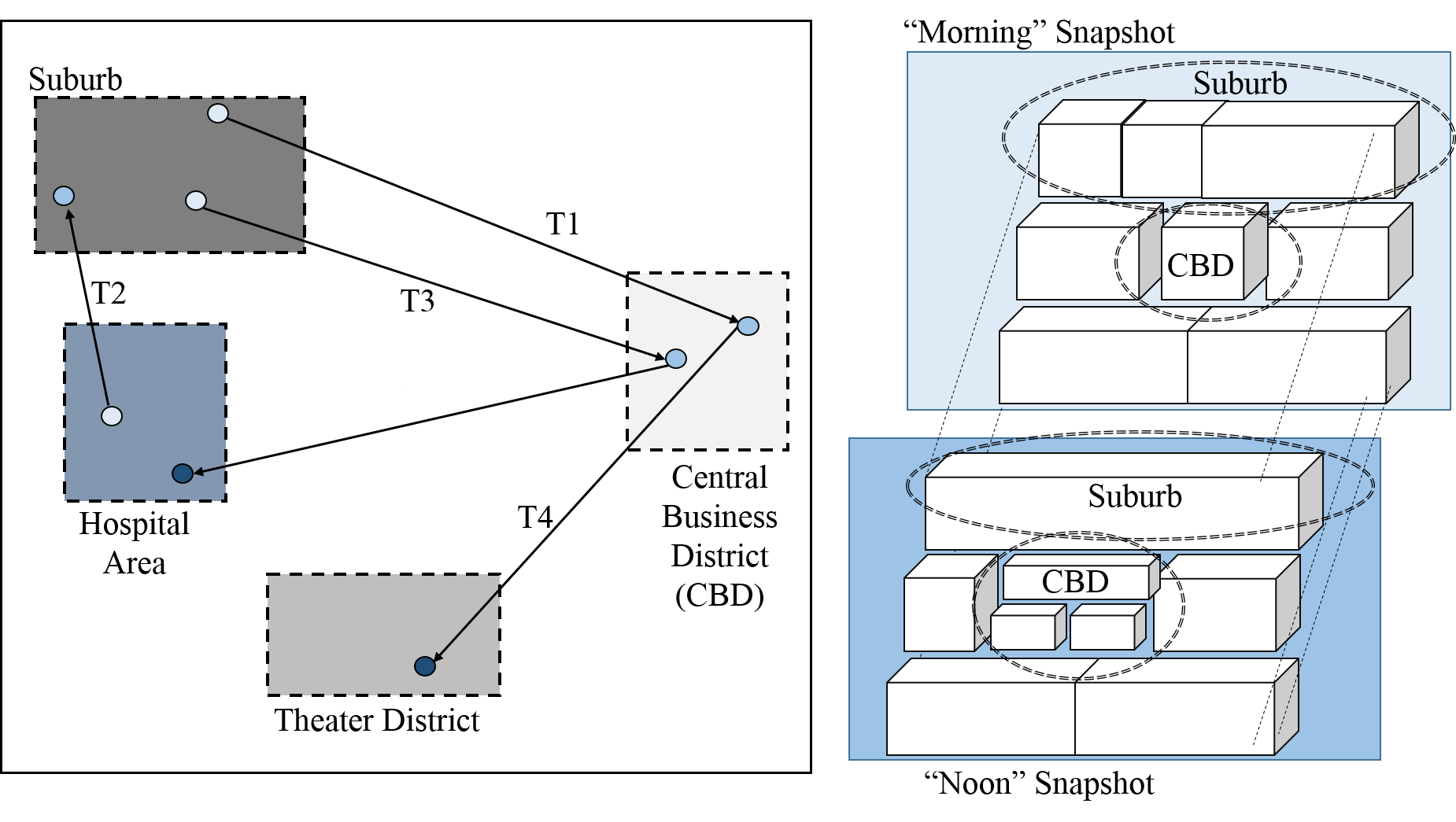}
	\caption{Capturing trajectory data using OD matrices.}
	\label{fig:OD}
\end{figure}

An important advantage of this representation is that the specific partitioning used for a particular dimension is customized to the data corresponding to that time frame. For instance, the same part of the space can be present in different frames, but with different granularities. In this example, the CBD area has low granularity for the morning time frame, since few people live there, but high granularity in the noon frame. Similarly, a theater district will not present interest in queries for the first or second time frame, but will likely be of high interest in the evening frame. Conventional OD matrices cannot accommodate such scenarios.


\section{Data-Independent Approaches} \label{Sec: Data Independent Approaches}
In this section, we introduce two data-independent approaches for the sanitization of frequency matrices with arbitrary dimensionality. These are extensions of existing work, particularly the technique in~\cite{AG}. In Section~\ref{Sec: Data Dependent Approaches} we will introduce more advanced data-dependent techniques that account for data distribution.

\subsection{Extended Uniform Grid (EUG)}~\label{Sec: EUG}
We extend the work in~\cite{AG},  originally proposed for two-dimensional frequency matrices. We refer to that algorithm as Uniform Grid (UG). The main idea of UG was to sanitize the total count of the frequency matrix and substitute it in a formula that results in a constant value $m$ that represents the granularity of dividing each dimension of a 2D frequency matrix.  After partitioning, the count in each of the partitions is sanitized using the Laplace mechanism. 

While the approach in~\cite{AG} only works for two-dimensional data, EUG provides a detailed analytical model that finds the optimal $m$ value for uniform partitioning in any number of dimensions. EUG is formally presented in Algorithm~\ref{Algo: EUG}. Suppose that the frequency matrix $F$ has $d$ dimensions represented by a $F_1\times F_2 \times ...\times F_{d}$ array, and let $N$ denote the total count of $F$. The objective is to find a value of $m$ such that, by updating the granularity of $F$ to $m^{d}$ and applying the Laplace mechanism, the utility of the published private histogram is maximized. The algorithm starts by utilizing a small amount of budget denoted as $\epsilon_0$ to obtain a noisy count of the total number of entries in the frequency matrix. 
\begin{equation}
    \overline{N} = N + Lap(s/\epsilon_0),
\end{equation}
where $\overline{N}$ denotes the sanitized count. The sanitized count is used for the estimation of $m$ by formulating an optimization problem. 

The value of $m$ can be estimated by considering the existing error sources, i.e., noise error and non-uniformity error. The former is used for sanitization of counts, and the latter is due to the assumption that data in each partition are uniform. Consider a query that selects $r$ portions of $F$, calculated by dividing its covered entries over the total number of entries. Hence, the query entails $rm^{d}$ entries of $F$. On the one hand, given that the noise added to each partition has a variance of $2/\epsilon^2$, the total additive noise variance sums up to $\dfrac{2rm^{d}}{\epsilon^2}$, or equivalently standard deviation of $\dfrac{\sqrt{2r}m^{d/2}}{\epsilon}$. \\

On the other hand, the query can be seen as a $d$-orthotope where the side length is proportional to $\sqrt[d]{r}$. Thus, each side of the orthotope spans $\sqrt[d]{r}\times m$ cells, and the number of points located inside the query is on average $\sqrt[d]{r}\times m \times \dfrac{ \overline{N}}{m^{d}}$. The term $ \overline{N}/m^{d}$ comes from the assumption of data uniformity in $F$. By further assuming that the non-uniformity error on average is some portion of the total density of the cells on the query border, we have the non-uniformity error as $\sqrt[d]{r}\times \dfrac{ \overline{N}}{c_0 m^{d-1}}$ for some constant $c_0$. Therefore, the aim is to find the optimal value of $m$ that minimizes the summation of two errors, i.e.,
\begin{equation}
   \min_{m} \dfrac{\sqrt{2r}m^{d/2}}{\epsilon} + \sqrt[d]{r}\times \dfrac{\overline{N}}{c_0 m^{d-1}}
\end{equation}
Solving based on stationary conditions of the above convex problem results in the optimal $m$ given by:

\begin{align}
   &\dfrac{d\sqrt{2r}m^{d-1}}{\epsilon} -(d-1) \sqrt[d]{r}\times \dfrac{\overline{N}}{c_0 m^{2d}}=0 \\
   & \rightarrow m = (\dfrac{2(d-1)}{d}\times r^{(1/d - 1/2)}\times \dfrac{\overline{N}\epsilon}{\sqrt{2}c_0 })^{2/(3d-2)}\label{Equ: m}
\end{align}
The base case of the problem occurs when the frequency matrix has two dimensions and results in the same equation proposed by~\cite{AG}:
\begin{equation}
    m=\sqrt{\dfrac{\overline{N}\epsilon}{\sqrt{2}c_0}}
\end{equation}
For higher dimensions, if query size is known in advance, Equation~(\ref{Equ: m}) can be used with the given $r$ to estimate the value of $m$; otherwise, by assuming that all query sizes are equally likely, integration over $r$ leads to Equation ($\ref{Equ: key formula}$). For derivation, let us define an auxiliary variable $\alpha$ as
\begin{equation}
    \alpha = (\dfrac{2(d-1)}{d}\times \dfrac{\overline{N}\epsilon}{\sqrt{2}c_0 })^{2/(3d-2)}
\end{equation}
Integration over $r$ leads to
\begin{align}
     \int_{0}^{1} \alpha \times r^{\dfrac{2-d}{d(3d-2)}} \,dr  
     &= \dfrac{\alpha}{\dfrac{2-d}{d(3d-2)}+1} r^{\dfrac{2-d}{d(3d-2)}+1} \Biggr|_{0}^{1}\\
     &= \alpha \times (\dfrac{d(3d-2)}{3d^2 -3d +2}),
\end{align}
and ultimately, results in:
\begin{equation} \label{Equ: key formula}
    m = (\dfrac{2(d-1)}{d}\times \dfrac{\overline{N}\epsilon}{\sqrt{2}c_0 })^{2/(3d-2)}\times (\dfrac{d(3d-2)}{3d^2 -3d +2}).\\
\end{equation}

Once the value of $m$ is calculated, each dimension of matrix $F$ is divided into $m$ equal intervals generating $m^d$ partitions. The entries in each partition are set to the partition's sanitized total count divided by the number of entries it contains. The sanitized total count of a partition is generated by adding its entries and using Laplace mechanism with the privacy budget of $\epsilon_{\text{tot}}-\epsilon_{\text{0}}$.


\begin{algorithm}[t]
\caption{Extended Uniform Grid (EUG)}\label{Algo: EUG}
\begin{flushleft}
    \hspace*{\algorithmicindent} \textbf{Input}:\;\; $F,\, \epsilon_{\text{tot}},\, \epsilon_0,\, s$;  
\end{flushleft}
\begin{algorithmic}[1]
\State $\overline{N} \leftarrow SUM(F) + Lap(s/\epsilon_0)$
\State $\epsilon_{\text{tot}}\leftarrow \epsilon_{\text{tot}}- \epsilon_0$
\State $d\leftarrow \text{Number of dims in } F$
\State $m \leftarrow (\dfrac{2(d-1)}{d}\times r^{(1/d - 1/2)}\times \dfrac{\overline{N}\epsilon}{\sqrt{2}c_0 })^{2/(3d-2)}$\label{line calculating m}
\State // UPDATE GRANULARITY
\State Divide each dimension by $m$
\For {each new partition $i$}
\State $N' \leftarrow$ $SUM(i)$ 
\State $\overline{N}' \leftarrow N' + Lap(s/\epsilon_{\text{tot}})$
\For {each entry $j$ in $i$}
\State $j\leftarrow \overline{N}'/|i|$
\EndFor
\EndFor
\State \textbf{return} $F$
\end{algorithmic}
\end{algorithm}

\subsection{Entropy-based Partitioning (EBP)}\label{Sec: prt-entropy}

A critical point in the EUG algorithm is how to determine the value of $m$. We propose Entropy-based Partitioning (EBP), a method for estimating a good value of $m$ based on the concept of entropy. In addition to providing better accuracy, EBP also addresses the issue with EUG's arbitrary choice of constant $c_0$ which is empirically set to $10/\sqrt{2}$. EBP proposes a more informed parameter selection process that does not require arbitrary value settings.


Consider a $d$-dimensional frequency matrix $F$ with dimensions $F_1\times F_2 \times ...\times F_{d}$, and let $N$ represent the total count of $F$. Moreover, denote the privacy budget allocated for the calculation of $m$ by $\epsilon$. As in the case of Algorithm~\ref{Algo: EUG}, the objective is to find a value of $m$ that, by updating the granularity of $F$ to $m^{d}$, and applying the Laplace mechanism, the utility of the published private histogram is maximized. We look at the problem from an information theory perspective. Once the granularity of $F$ is updated, the variance of total Laplace noise used to sanitize partitions adds up to $\dfrac{2m^{d}}{\epsilon^2}$, leading to total standard deviation of $\dfrac{\sqrt{2}m^{d/2}}{\epsilon}$. The entropy of the noise imposed on the frequency matrix is therefore,
\begin{equation}
    H(\dfrac{\sqrt{2}m^{d/2}}{\epsilon}) = -\log_2{\dfrac{\epsilon}{\sqrt{2}m^{d/2}}}.
\end{equation}
On the other hand, consider the amount of information loss that occurs due to the change in granularity. To calculate the information loss, the amount of information before and after changing the granularity $F$ is required. The information contained in $F$ before change of granularity can be calculated as $H(F)$, denoting the entropy of $F$. After partitioning is conducted, the entropy is reduced to $H(F|m)$, denoting entropy calculated based on the updated frequency matrix with the granularity of $m^d$. Thus, the amount of information loss incurred due to change in granularity is:
\begin{equation}
    \text{Information Loss} = H(F) - H(F|m).
\end{equation}
An optimization problem can be formulated to find the optimal value of $m$ that minimizes the average query error. 
\begin{equation}
    \min_{m} H(\dfrac{\sqrt{2}m^{d/2}}{\epsilon})+ H(F) - H(F|m).
\end{equation}
By increasing the value of $m$, information loss becomes smaller, but the induced noise grows larger. The optimal value of $m$ is reached when the noise is equal to information loss. Unfortunately, entropy cannot be directly calculated due to privacy concerns; however, an approximation can be employed as follows. We assume that the number of entries is in the order of the number of data points, and data points are uniformly distributed over the $m^d$ partitions. Entropy before/after changing granularity can be approximated as
\begin{equation}
    H(F) \approx -\log_2 (1/N),\, H(F|M) \approx -\log_2 (1/m^{d})
\end{equation}
To preserve the privacy of users, the value of $N$ is sanitized beforehand based on the Laplace mechanism. The value of $m$ minimizing the optimization problem is derived as 
\begin{align}
  &-\log_2{\dfrac{\epsilon}{\sqrt{2}m^{d/2}}}  = -\log_2 (1/N)+\log_2 (1/m^{d})\\
  &\rightarrow \log_2{\dfrac{\epsilon}{\sqrt{2}m^{d/2}}} = \log_2 (m^{d}/N) 
  \rightarrow m = \sqrt[3d/2]{\dfrac{N\epsilon}{\sqrt{2}}}\label{Equ: entropy formula}
\end{align}

The derived formula in Equation (\ref{Equ: entropy formula}) is an alternative method to calculate the value of $m$ in the EUG algorithm. Therefore, the pseudocode in Algorithm~\ref{Algo: EUG} applies to EBP by replacing the formula in line~\ref{line calculating m} with Equation (\ref{Equ: entropy formula}).

\section{Data-Dependent Approaches}\label{Sec: Data Dependent Approaches}

\subsection{Overview}\label{subsection: Overview}

Data-independent algorithms overlook critical information about the distribution of data points, as they always assume uniform distribution. This is particularly problematic for higher dimensional frequency matrices, due to their tendency to be sparse. 

To improve accuracy when publishing higher dimensional frequency matrices, we propose a tree-based approach called Density-Aware Framework (DAF) that takes into account density variation across different regions of the space. In addition, DAF introduces a key feature that enables custom stop conditions for partitioning. Intuitively, denser parts of the space should be split in more granular fashion, while for sparse areas the partitioning can stop earlier, since most likely large regions of the space are empty. The decision of when to stop partitioning the frequency matrix is made privately, and avoids over-partitioning which can lead to large errors in higher dimensional frequency matrices. 

DAF is a hierarchical partitioning approach that resembles a tree index. Each node covers a portion of the frequency matrix, with the root node covering all entries. Descendants of a node are generated by a non-overlapping split of the parent node's entries. The split is conducted based on the depth of the node, such that nodes at depth $i$ are created by dividing dimension $i$ of their parent node's partition. The maximum index height is $d+1$. The fanout and the split point are customized at each node based on sanitized local information about the data. We propose two DAF alternatives based on different split objective functions: (i) {\em DAF-Entropy} (Section~\ref{Sec: DAF-Entropy}) which uses entropy information to estimate good split parameters, and (ii)  {\em DAF-Homogeneity} (Section~\ref{DAF-Homogeneity}) which focuses on creating partitions with high intra-region homogeneity. Section~\ref{Sec: budget allocation} introduces privacy budget allocation considerations that are relevant to both approaches.

\begin{algorithm}[t]
\caption{DAF-Entropy}\label{Algo: DAF-Entropy}
\begin{algorithmic}[1]
\State \textbf{Global Constants:} $\epsilon_{\text{tot}},\, m_0$
\Function{DAF-Entropy}{$ x,\, acc$}
\State $d\leftarrow$ Number of dimensions
\State $d'\leftarrow x.depth$
\If{$d'= d$ }
    \State $x.\text{ncount} \leftarrow x.\text{count} + Lap(1/(\epsilon_{\text{tot}}-acc))$
    \State \textbf{return} TRUE
\EndIf
\If{$d' = 0$}
    \State x.ncount = x.count + $Lap(1/(\epsilon_{tot}/100))$
    \State $acc \leftarrow acc+ \epsilon_{tot}/100$
    \State $m_0,m \leftarrow \sqrt[3(d-d')/2]{\dfrac{(x.ncount)\times(\epsilon_{\text{tot}}-acc)}{\sqrt{2}}} $
\Else
    \State $mem \leftarrow 
        \dfrac{\epsilon_{\text{tot}}\times m_0^{d'/3} \times (1-m_0^{1/3})}{m_0^{1/3} (1- m_0^{d/3}) } $ \label{l1}
    \State $x.\text{ncount} \leftarrow x.\text{count} + Lap(1/mem)$\label{l2}
    \State $acc \leftarrow acc+mem$
    \State $m \leftarrow \sqrt[3(d-d')/2]{\dfrac{(x.ncount)\times(\epsilon_{\text{tot}}-acc)}{\sqrt{2}}}$    
\EndIf
\If{Stop Conditions$\Equal$ TRUE}
    \State $mem \leftarrow \epsilon_{\text{tot}} -acc$
    \State $x.\text{ncount} \leftarrow x.\text{count} + Lap(1/mem)$
    \State \textbf{return} TRUE
\EndIf 
\State $M\leftarrow$ Split $(d'+1)$th dimension into $m$ intervals 
\For {$i\Equal 1\, \text{to}\, m$}
\State create a new node $x'$
\State $x'.F \leftarrow x.F$ with $i$th dimension set to $M[i]$
\State $x'.depth \leftarrow d'+1$ 
\State $x'.count \leftarrow \, SUM(x'.F)$ 
\State DAF-Entropy($ x',\, acc$)
\EndFor
\EndFunction
\end{algorithmic}
\end{algorithm}

\subsection{DAF-Entropy}~\label{Sec: DAF-Entropy}

DAF-Entropy has the recursive structure presented in Algorithm~\ref{Algo: DAF-Entropy}. It receives as inputs the current node to split denoted by $x$, privacy budget $\epsilon_{\text{tot}}$, variable $acc$ tracking the budget spent so far (initially set to zero), and a constant $m_0$ set in the first round of the recursion which is used for budget allocation purposes at all levels of the tree (more details are provided in Section~\ref{Sec: budget allocation}). Each tree node $x$ is an object with four attributes: (i) $x.F$; the node's associated entries in the frequency matrix, (ii) $x.count$; the actual sum of entries in $x.F$, (iii) $x.ncount$; the sanitized (or noisy) count, and (iv) $x.depth$; the node's depth in the tree. The initial run of the function is performed for the root node, representing the whole frequency matrix.

DAF-Entropy sanitizes the total count of the root node and utilizes Equation~(\ref{Equ: entropy formula}) to partition the first dimension of the frequency matrix. New nodes are generated for each new partition assigned as one of the node's children. The algorithm recursively visits children and repeats the same process with the key difference that the split is done based on the second dimension. More generally, upon reaching a node at depth $i$, the split is conducted in the $(i+1)$-th dimension.

Once a new node is visited, its count is sanitized, and stop conditions are tested on the sanitized count. If satisfied, the tree is pruned, and the node turns into a leaf. A special technique is used in such a scenario to enhance accuracy. The algorithm uses the remaining amount of budget that was supposed to be used while visiting children to update the sanitized count. This technique improves accuracy as budget allocation is such that lower levels of the tree are allocated more budget. Thus, it is worth updating the sanitized count based on the remaining amount of budget. Note that, stop conditions can be selected based on application-specific details; however, the most prominent stop condition that can help avoid over-partitioning is to stop when the sanitized count is below a certain threshold. The algorithm continues until reaching depth $d$, indicating that partitioning on all $d$ dimensions has been implemented successfully or a stop condition is reached. Finally, the sanitized counts of the leaves are published, representing the private frequency matrix. 


\begin{algorithm}[t]
\caption{DAF-Homogeneity}\label{Algo: DAF-Homogeneity}
\begin{algorithmic}[1]
\State \textbf{Global Constants:} $p,\, q,\, \epsilon_{\text{tot}},\, m_0$
\Function{DAF-Homogeneity}{$ x,\, acc$}
\State $d\leftarrow$ Number of dimensions
\State $d'\leftarrow x.depth$
\If{$d'= d$ }
    \State $x.\text{ncount} \leftarrow x.\text{count} + Lap(1/(\epsilon_{\text{tot}}-acc))$
    \State \textbf{return} TRUE
\EndIf
\If{$d' = 0$}
    \State x.ncount = x.count + $Lap(1/(\epsilon_{tot}/100))$
    \State $acc \leftarrow acc+ \epsilon_{tot}/100$
    \State $m_0,m \leftarrow \sqrt[3(d-d')/2]{\dfrac{(x.ncount)\times(\epsilon_{\text{tot}}-acc)}{\sqrt{2}}} $
    \State $\mathcal{K}_1,\,... \mathcal{K}_p \leftarrow$ Use $m$ to generate candidate sets \label{r1}
    \State Compute $O(\mathcal{K}_1),\, O(\mathcal{K}_2),.., O(\mathcal{K}_p)$
    \State $\overline{O(\mathcal{K}_i)} \leftarrow O(\mathcal{K}_i) + Lap(2/(p\times \epsilon_\text{prt})),\forall i=1... p$
    \State $\mathcal{K}\leftarrow   \underset{i}{Minimize}\;\; \overline{O(\mathcal{K}_i)} \;\;\;  \forall i \Equal 1...p$\label{r2}
\Else
    \State $\epsilon \leftarrow 
        \dfrac{\epsilon_{\text{tot}}\times m_0^{d'/3} \times (1-m_0^{1/3})}{m_0^{1/3} (1- m_0^{d/3}) } $ 
    \State $acc \leftarrow acc+\epsilon$
    \State $\epsilon_{\text{prt}} \leftarrow q\epsilon$ 
    \State $\epsilon_{\text{data}} \leftarrow  (1-q)\epsilon$    
    \State $x.\text{ncount} \leftarrow x.\text{count} + Lap(1/\epsilon_{\text{data}})$
    \State $m \leftarrow \sqrt[3(d-d')/2]{\dfrac{(x.ncount)\times(\epsilon_{\text{tot}}-acc)}{\sqrt{2}}}$
    \State Execute lines \ref{r1} to \ref{r2}
\EndIf
\If{Stop Conditions$\Equal$ TRUE}
    \State $x.\text{ncount} \leftarrow x.\text{count} + Lap(1/(\epsilon_{\text{tot}} -acc))$
    \State \textbf{return} TRUE
\Else
\State $M\leftarrow$ Split $(d'+1)$th dimension based on $\mathcal{K}$ 
\For {$i\Equal 1\, \text{to}\, m$}
\State create a new node $x'$
\State $x'.F \leftarrow x.F$ with $i$th dimension set to $M[i]$
\State $x'.depth \leftarrow d'+1$ 
\State $x'.count \leftarrow \, SUM(x'.F)$ 
\State DAF-Homogeneity($ x',\, acc$)
\EndFor
\EndIf 
\EndFunction
\end{algorithmic}
\end{algorithm}

\subsection{DAF-Homogeneity}\label{DAF-Homogeneity}

The partitioning process plays a critical role in the private publication of frequency matrices. Hence, several attempts have been made in prior work~\cite{xiao2010differentially,ashwin} to find an efficient splitting mechanism, including partitioning independent of data, based on medians or using the frequency matrix's total count to estimate a viable partitioning granularity. Our earlier work in~\cite{shaham2021htf} shows that partitioning based on homogeneity can significantly improve the utility of private frequency matrices in 2D. The principal idea is to have mechanisms that can cluster the entries such that data density is homogeneous within each cluster. Recall that partitioning needs to follow the DP constraint as with any other part of the algorithm. Here, we extend the approach in~\cite{shaham2021htf} to higher dimensional frequency matrices. The approach is built on top of Algorithm~\ref{Algo: DAF-Entropy}, with a key difference that once fanout is calculated for a node, an alternative method is used to partition the space based on homogeneity. 

Suppose that while executing Algorithm~\ref{Algo: DAF-Entropy}, a node with depth $i$ is visited. DAF-Homogeneity starts by dividing the calculated amount of budget into two parts: sanitization budget ($\epsilon_{data}$), and partitioning budget ($\epsilon_{\text{prt}}$).  
\begin{equation}
    \epsilon_{\text{prt}} = q\epsilon_i,\, \epsilon_{\text{data}} = (1-q)\epsilon_i
\end{equation}
Constant $q$ denotes the ratio of the budget assigned for partitioning. This value is experimentally set to $0.3$. Next, the node's count is sanitized based on the Laplace mechanism with the privacy budget $\epsilon_{data}$, and substituted in Equation~(\ref{Equ: entropy formula}) to calculate the fanout $m$.

Suppose that $m=3$ and recall that for nodes with depth $i$, the split is conducted on dimension $i+1$. Let us denote the interval corresponding to the $(i+1)$-th dimension by $[k_{\text{start}},k_{\text{end}}]$. In the case of DAF-Homogeneity, given that the fanout is calculated to be $3$, the generated intervals for the $i+1$ dimension of children would be $[k_{\text{start}},k_1)$, $[k_1, k_2)$, and $[k_2, k_{\text{end}}]$, where
\begin{equation}
    k_1 = \lfloor k_0 + \dfrac{k_{\text{end}}- k_{\text{start}}}{3}\rfloor,
    k_2 = \lfloor k_0 + 2\times\dfrac{k_{\text{end}}- k_{\text{start}}}{3}\rfloor
\end{equation}

Instead of simply selecting $k_1$ and $k_2$ as splitting points, DAF-Homogeneity follows an alternative method: it generates $p$ sets of candidate partitioning sets $\mathcal{K}_1,\, \mathcal{K}_2,\,... \mathcal{K}_p$, where $p$ is an input to the algorithm. Each set $\mathcal{K}_j$ 
has a cardinality equal to the desired fanout, and is generated by drawing uniformly random split positions from every partition. For example, consider the first candidate set to be  $\mathcal{K}_1 = \{ k_1',k_2',k_3'\}$, where $k_1',k_2',\text{and } k_3'$ are uniformly random coordinates drawn from intervals $[k_{\text{start}},k_1)$, $[k_1, k_2)$, and $[k_2, k_{\text{end}}]$, respectively. Furthermore, let us denote the frequency matrix generated by setting the $i+1$ dimension into the $j$th interval by $F^j$. Next, the algorithm computes the homogeneity objective function for candidate sets, resulting in  $O(\mathcal{K}_1),\, O(\mathcal{K}_2),.., O(\mathcal{K}_p)$, where 
\begin{align}\label{Equ: objective function}
    O(\mathcal{K}) = \sum_{i=1}^{|\mathcal{K}|+1} \sum_{f_j\in F^i}|f_{j} - \mu_{F^i}|,
\end{align}
In the above equation, $\mu_{F^i}$ denotes the average of entries in $F^i$.
\begin{equation}
   \mu_{F^i} = \dfrac{\sum_{f_j\in F^i}f_{j} }{|F^i|}
\end{equation}
Then, the output values are sanitized based on the Laplace mechanism with the reserved privacy budget for partitioning. 
\begin{equation}
\overline{O(\mathcal{K}_i)} = O(\mathcal{K}_i) + Lap(s/\epsilon_\text{prt}),\forall i=1... p
\end{equation}
The optimal candidate set is chosen as the one that results in the minimum sanitized output. 
\begin{equation}
   \underset{i}{Minimize}\;\; \overline{O(\mathcal{K}_i)} \;\;\;  \forall i \Equal 1...p
\end{equation}

\begin{lemma}\label{lemma: sanitize}
Sensitivity of the homogeneity objective function is $2$.
\end{lemma}
\begin{proof}
In the calculation of objective function $O(\mathcal{K})$  for a given split index set $\mathcal{K}$, a data entry's existence or absence only affects one cell and the corresponding cluster. Let us denote the objective function after addition or removal of one data record by $O(\mathcal{K})'$. 
\begin{align}
     O(\mathcal{K})' = \sum_{i=1}^{|\mathcal{K}|+1} \sum_{f_j'\in F^i}|f_{j}' - \mu_{F_i}'|,
\end{align}

Without loss of generality assume that the additional record is located in the first cluster which results in $\mu_{F^1}' = \mu_{F^1} + 1/|F^1|$, and $\mu_{F^i}' = \mu_{F^i}$ for all $i=2,...,k+1$. Similarly, the counts are equal ($f_i' = f_i $) for all entries except a single entry denoted by $x$ for which we have we have $f_{x}' = f_{x} + 1$. Writing triangle inequality results in
\begin{equation}
\begin{split}
   \Big{|} |f_{i} - \mu_1 - \dfrac{1}{|F^1|}|-   |f_{i} - \mu_1| \Big{|} &\leq \dfrac{1}{|F^1|}\\ &\forall i\Equal1...|F^1|-\{x\}
\end{split}
\end{equation}
The sensitivity of the objective function can have the maximum value of two, as proven by the following inequality:
\begin{equation}
  \Big{|} | f_{x} +1- \mu_1 - \dfrac{1}{|F^1|}|-   |f_{x} - \mu_1|  \Big{|}\leq \dfrac{|F^1|-1}{|F^1|} 
\end{equation}
\end{proof}
The DAF-Homogeneity  pseudocode is shown in Algorithm~\ref{Algo: DAF-Homogeneity}.

\begin{figure}[!t]
	\subfloat[Non-adaptive approaches \label{fig_ex1}]{%
	\includegraphics[scale=.18]{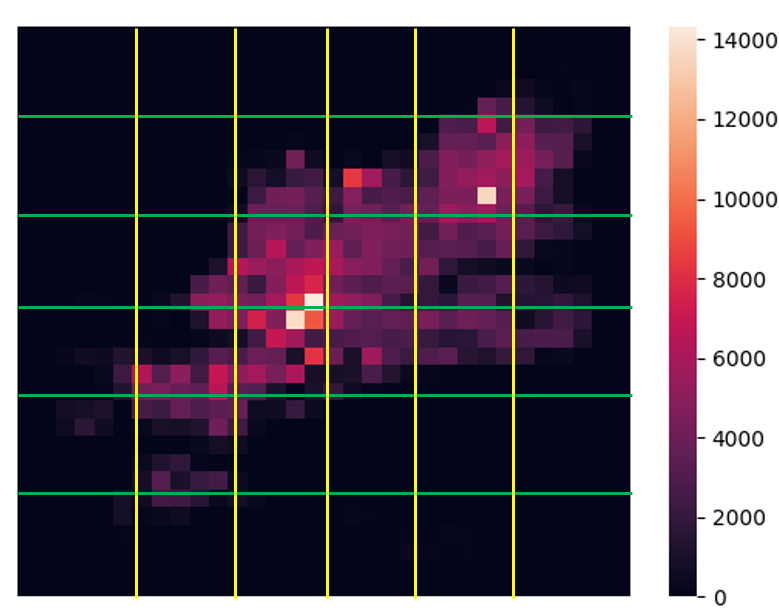}
	}
	\hfill
	\subfloat[DAF-Entropy \label{fig_ex2}]{%
	\includegraphics[scale=.18]{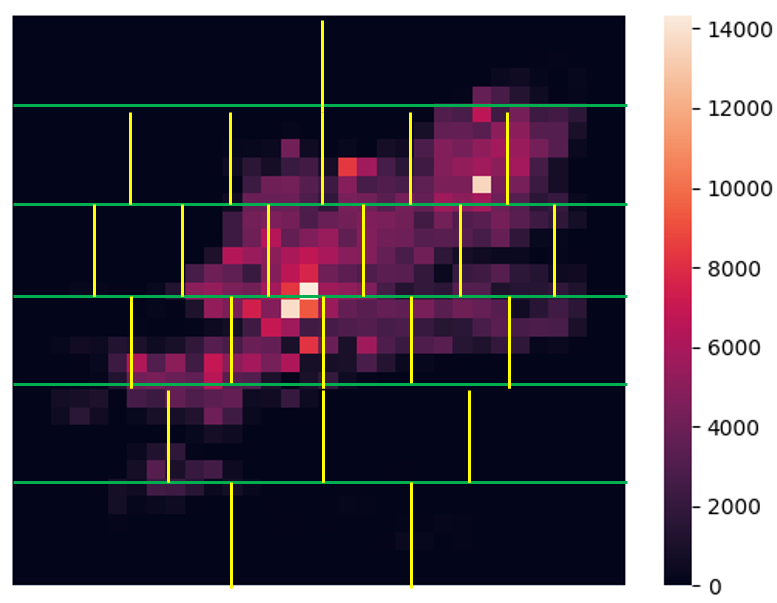}
	}
	\hfill
	\centering
	\subfloat[DAF-Homogeneity\label{fig_ex3}]{%
	\includegraphics[scale=.18]{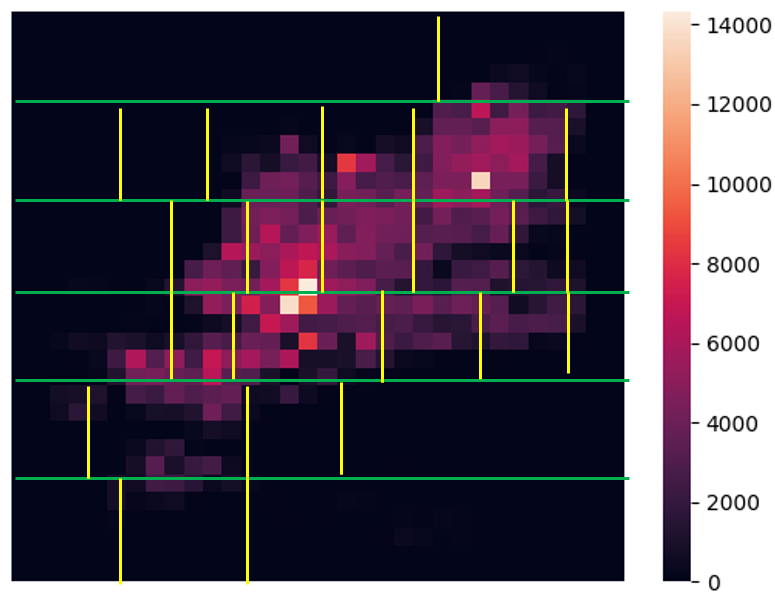}
	}
	\caption{Intuition behind DAF sanitization approaches.}
	\label{fig_ex}
\end{figure}

To better understand the reason why the proposed DAF approaches outperform competitor techniques, Fig.~\ref{fig_ex} provides a heatmap representation of Los Angeles city with 500,000 sampled from the Veraset dataset (experimental setup details are provided in Section 6.1). The partitioning conducted in the first and second dimensions are shown by green and yellow lines, respectively. For non-adaptive approaches, only the sanitized total population count is used for partitioning, and therefore, both dimensions are divided equally without considering user distribution (Fig.~\ref{fig_ex1}). Conversely, the DAF-Entropy approach is adaptively adjusting the number of partitioning as the dimension changes (Fig.~\ref{fig_ex2}). The DAF-Homogeneity technique goes one step further, and adjusts the number of partitions generated in each dimension, by selecting the split point such that resulting areas will exhibit homogeneous intra-bin density, hence reducing the negative effects of uniformity assumption and increasing query accuracy.

\subsection{Budget Allocation}\label{Sec: budget allocation}

The derivation of the optimal amount of privacy budget allocated for different levels of the hierarchy is a challenging task as nodes have varying fanouts. We formulate an optimization problem to achieve a good quality budget allocation. Denote the fanout of the root node by $m_0$. We assume that the progression of fanout is geometric. At depth $i$, there exist approximately $m_0^i$ nodes. Furthermore, we denote the budget allocated to depth $i$ of the tree by $\epsilon_i$. The goal is to minimize the variance of the noise added to each level:
\begin{align}
    \underset{\epsilon_1...\epsilon_d}{min}\;\;\;& \sum_{i\Equal1}^d m_{\text{0}}^i/\epsilon_i^2,
    \sum_{i\Equal1}^d \epsilon_i=\epsilon_{\text{tot}}', \;\; \epsilon_i>0 \; \;\forall i=1...d
\end{align}
where, $\epsilon_{\text{tot}}' = \epsilon_{\text{tot}}-\epsilon_0$. We have intentionally separated the root node's budget, as it will be used to calculate $m_0$. The optimization problem can be solved by writing Lagrangian and KKT conditions.
\begin{align}
    L(\epsilon_1,...,\epsilon_d,\lambda) &= \sum_{i\Equal1}^d m_0^i/\epsilon_i^2+ \lambda ( \sum_{i\Equal 1}^d \epsilon_i- \epsilon)\\
    &\Rightarrow \dfrac{\partial L}{\partial\epsilon_i} =  \dfrac{-2m_{\text{0}}^i}{\epsilon_i^3}+ \lambda =0
    \Rightarrow \epsilon_i = \dfrac{(2m_{\text{0}}^i)^{1/3}}{\lambda^{1/3}},
\end{align}
which leads to 
\begin{equation}
    \epsilon_i = \dfrac{\epsilon_{\text{tot}}'\times
    m_{\text{0}}^{i/3}}{\sum_{i\Equal1}^d m_{\text{0}}^{i/3}} = 
    \dfrac{\epsilon_{\text{tot}}'\times m_{\text{0}}^{i/3} \times (1-m_{\text{0}}^{1/3})}{m_{\text{0}}^{1/3} (1- m_{\text{0}}^{d/3}) }.
\end{equation}

A question arises on how to calculate the value $m_0$ upon which the above optimization problem is formulated. Note that the formulation only considers depths $1$ to $d$, and the root node is excluded from the equation. The value of $m_0$ is calculated in the first run of the recursive algorithm~\ref{Algo: DAF-Entropy}, and we set the budget to:
\begin{equation}
    \epsilon_0 =\dfrac{\epsilon_{\text{tot}}}{100}
\end{equation}
Therefore, a comparably small amount of budget is allocated to the root node to derive $m_0$. Based on the above formulation, one can see that lower levels of the tree benefit from significantly higher levels of budget. This helps to improve the utility of the published private histogram, as the sanitized leaf set of the tree represents the counts published by our approach.

\section{Related Work}\label{Sec: literature reivew}

Prior works on private publication of frequency matrices can be classified into three categories: data independent, partially data dependent, and data dependent algorithms. The algorithms in the first category are independent of the underlying dataset. The partial data dependent algorithms are the category of algorithms where the number of data points is used to generate the private FMs, but no consideration is made for the data distribution. The algorithms in the last category take the distribution of data points into consideration to improve the utility. Most algorithms are developed to address only the publication of 1D and 2D FMs.

In the category of data-independent approaches, two baseline algorithms that stand out are called {\em singular} and {\em identity}. The singular algorithm~\cite{ashwin} considers the frequency matrix as a single partition and adds Laplace noise to the total count. The queries are answered based on the sanitized total count only, considering the assumption of data uniformity. The identity algorithm~\cite{dwork2006calibrating} on the other hand, adds Laplace noise to each entry of the frequency matrix. The number of partitions in this algorithm is equal to the total number of entries. 
The {\em Privlet} algorithm~\cite{xiao2010differential} enhances the performance of the identity algorithm by transforming the frequency matrix based on wavelets and by adding  noise in the new domain. Then, the algorithm converts back to the noisy matrix and releases the DP counts. The authors in~\cite{cormode2012differentially} build a quadtree on top of the FM:  a tree-based spatial indexing structure that splits each node into four equal quarters, regardless of data placement. The so-called binning or partitioning of space without observing the histograms is studied in~\cite{cormode2021data}. The authors consider the amount of overlap between bins and propose an algorithm called 'varywidth' that provides improved performance in terms of the trade-off between the spatial precision and the accumulated variance over differentially private queries. The use of summaries for private publication of histograms is explored in~\cite{cormode2012differentially2}. The authors show it is possible to reduce the two-step approach of generating private summaries, in which first the private histogram is generated and then the summaries are released, to a one-step approach. The one-step method prevents the data owner and data user from getting overwhelmed with the large computational complexity overhead. 

In contrast to the data independent algorithms, data dependent approaches exploit the distribution of data in the FM to answer queries with higher accuracy. General purpose mechanisms~ \cite{li2013optimal, li2012adaptive} and their workload-aware counterpart DAWA \cite{li2014data} operate over a discrete 1D domain; however, they can be applied to the 2D domain by dimensional reduction transformations such as Hilbert curves~\cite{ashwin}. Unfortunately, dimensionality reduction can prevent range queries from being answered accurately, and also increases computational complexity. This significantly limits their practicality, particularly for higher-dimensional data. Data-aware tree-based algorithms such as k-d-trees~\cite{xiao2010differentially} allocate a portion of the budget to partitioning, and generate split points based on density. Hybrid approaches between data-independent and data-dependent algorithms have also been proposed, e.g., UG and AG~\cite{AG}. We refer to these approaches as partially data-dependent. Only the sanitized total count of the FM is used in the partitioning process.  The UG algorithm and its extension~\cite{AG} sanitize the total count of FMs and use it to alter the granularity of FM such that the utility of the published private FM is improved. The MKM approach proposed in~\cite{lei2011differentially} provides an alternative formula to partition FM considering its dimensionality. As is the case in UG, the formula only takes as input the total count of the frequency matrix and determines the granularity of FM based on the sanitized total count. In some cases, such approaches have been shown to provide superior performance to more complex methods~\cite{ashwin}.


There is prior work in storage, processing, and compression of histograms, but without considerations for privacy. The authors in~\cite{kernert2015spmacho} focus on lowering the computational complexity of matrix multiplication and storage. The proposed approach generates an execution plan for the multiplication of dense and sparse matrices. A cost model is also proposed to understand the sparsity of matrices and the estimation of density. The execution plan tends to optimize the overall cost overhead. An adaptive tile matrix representation is proposed in~\cite{kernert2016topology} for large matrix multiplication. An operator called ATMULT with the capability of shared memory parallel matrix multiplication is proposed for dynamic tile-granular optimizations, conducted based on the density estimation. The work in~\cite{diakonikolas2018fast} studies the problem of density estimation for higher dimensional histograms. The main idea is to estimate the distribution of data for a given set of samples. The algorithm provides near-optimal sample complexity, i.e. close to theoretical information limit, and runs in polynomial time.


\begin{center}
\begin{table}[t]
\caption{Summary of Compared Approaches}
\centering
\begin{tabular}{| >{\centering\arraybackslash}m{1cm} |  >{\centering\arraybackslash}m{3cm} ||  >{\centering\arraybackslash}m{3cm} |}
\hline        \multicolumn{2}{|c||}{Strategy}                     & Symbol                    \\
\hline        \multicolumn{2}{|c||}{Baseline Algorithms}            & IDENTITY      \cite{dwork2006calibrating}\\
   \multicolumn{2}{|c||}{}                              & UNIFORM        \cite{dwork2006calibrating}\\
\hline        \multicolumn{2}{|c||}{Non-adaptive Sanitization }      & EUG          \\
   \multicolumn{2}{|c||}{Approaches}                              & EBP  \\
   \multicolumn{2}{|c||}{}                             & MKM         \cite{lei2011differentially} \\
\hline \multicolumn{2}{|c||}{With partitioning budget} & DAF-Entropy \\
\hline \multicolumn{2}{|c||}{Without partitioning budget} & DAF-Homogeneity\\
\hline
\end{tabular}
\label{Table: contribution}
\end{table}
\end{center}


\section{Experimental Evaluation}\label{Experimental Evaluation}

\subsection{Experimental Setup} \label{Subsec: Experimental setup}

{\bf Synthetic Datasets.} We generate synthetic frequency matrices according to both Gaussian and Zipf distribution. To generate a $d$-dimensional Gaussian frequency matrix $F$ with dimensions $F_1\times F_2 \times ...\times F_d$, a uniformly random integer is sampled in each dimension: $c_i\sim \text{Uniform}(1,F_i), \;\forall i=1...d$. The generated point $(c_1,c_2,...,c_d)$ is selected as the cluster center and $1$ million datapoints are generated with respect to the cluster center according to a normal distribution. Specifically, each data point $(x_1,x_2,...,x_d)\in \mathcal{Z}^d$ is sampled from a multivariate Gaussian distribution $(X_1,X_2,...,X_d)$, where $X_i\sim \mathcal{N}(c_i,var)$. 
Changing the variance $var$ allows us to adjust the degree of data skewness (lower values of $var$ will correspond to more skewed data). Zipfian data are generated by sampling each datapoint from a multivariate Zipf distribution $(Y_1,Y_2,...,Y_d)$, where $Y_i\sim \dfrac{x^{-a}}{\zeta(a)}$. $\zeta(.)$ denotes the Riemann Zeta function
and parameter $a$ controls the skew in the frequency matrix. As opposed to variance in Gaussian distribution, a higher value of $a$ results in a more skewed distribution for the Zipf distribution.

{\bf Real-world datasets.} We use a subset of the Veraset\footnote{Veraset is a data-as-a-service company that provides anonymized population movement data collected through signals of cell phones across the USA.} dataset~\cite{datarade}, including location measurements of cell phones in three US cities: New York, Denver and Detroit. For each city, we consider a large geographical region covering a $70 \times 70$ km$^2$ area centered at the city's central latitude and longitude.
These are chosen to represent cities with high, moderate and low densities, respectively. 
Cities are modeled by a $1000\times 1000$ frequency matrix where each entry represents the number of data points in the corresponding region of the city. The selected data generates a frequency matrix of 1 million data points during the time period March 1-7, 2020.

\begin{figure*}[t]
	\subfloat[2D, $\epsilon_{\text{tot}} = 0.1$.\label{s1}]{%
	\includegraphics[scale=.43]{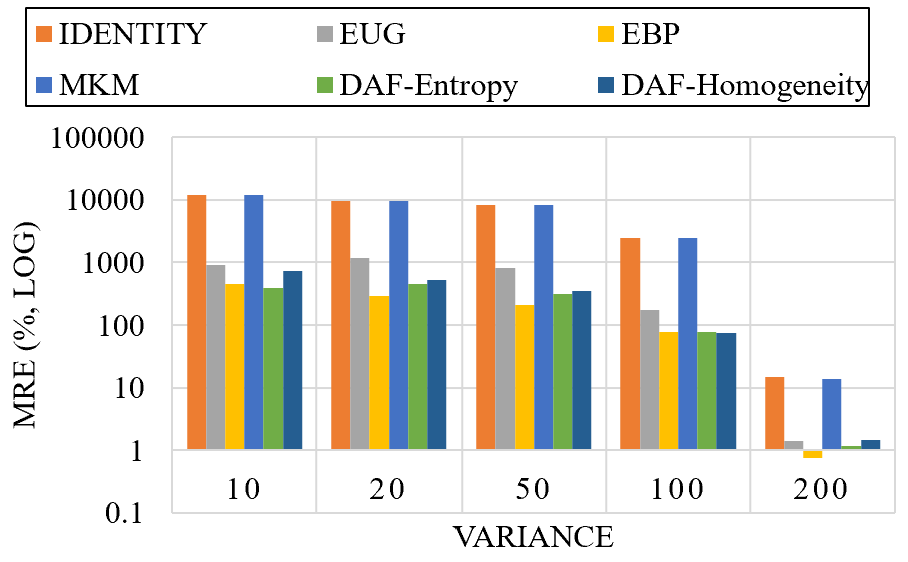}
	}
	\hfill
	\subfloat[2D, $\epsilon_{\text{tot}} = 0.3$. \label{s2}]{%
	\includegraphics[scale=.43]{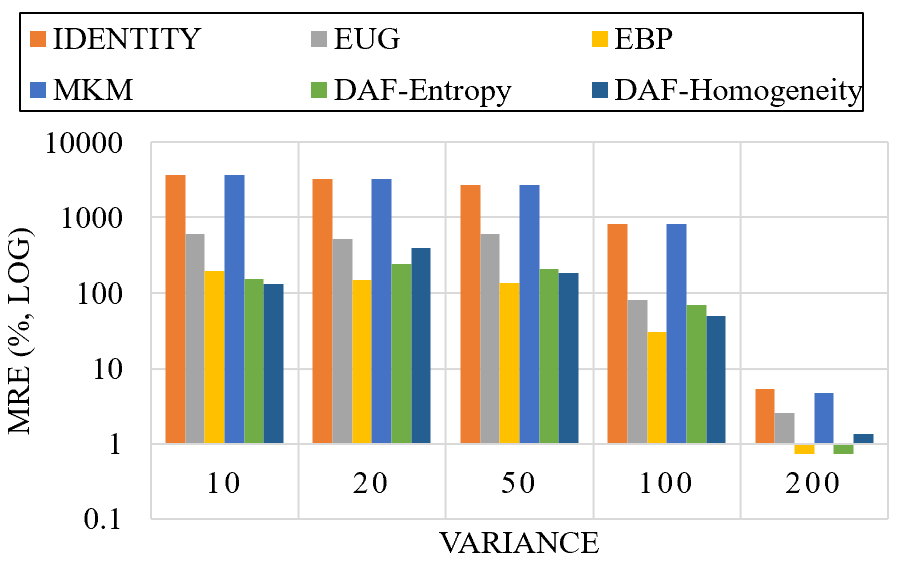}
	}
	\hfill
	\subfloat[2D, $\epsilon_{\text{tot}} = 0.5$. \label{s3}]{%
	\includegraphics[scale=.43]{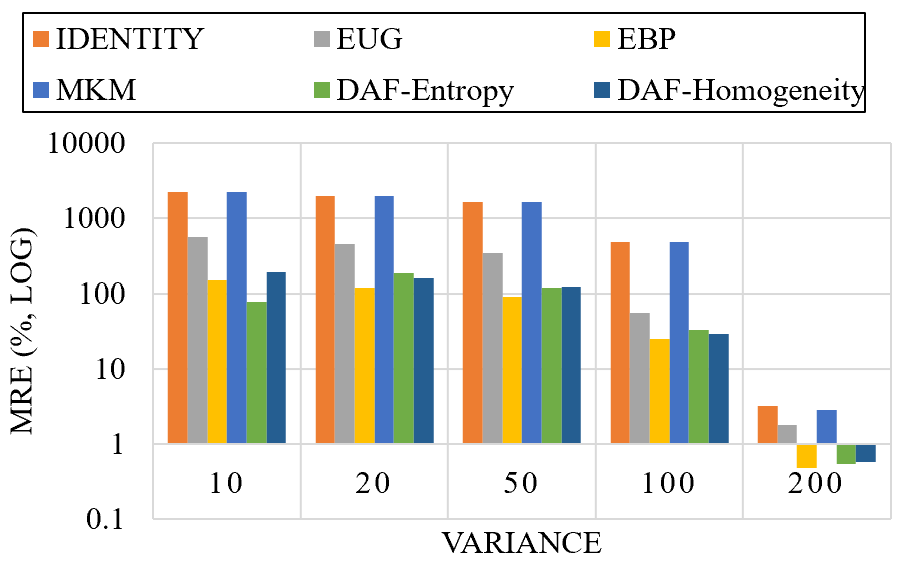}
	}
	\hfill
	\subfloat[4D, $\epsilon_{\text{tot}} = 0.1$.\label{s4}]{%
	\includegraphics[scale=.43]{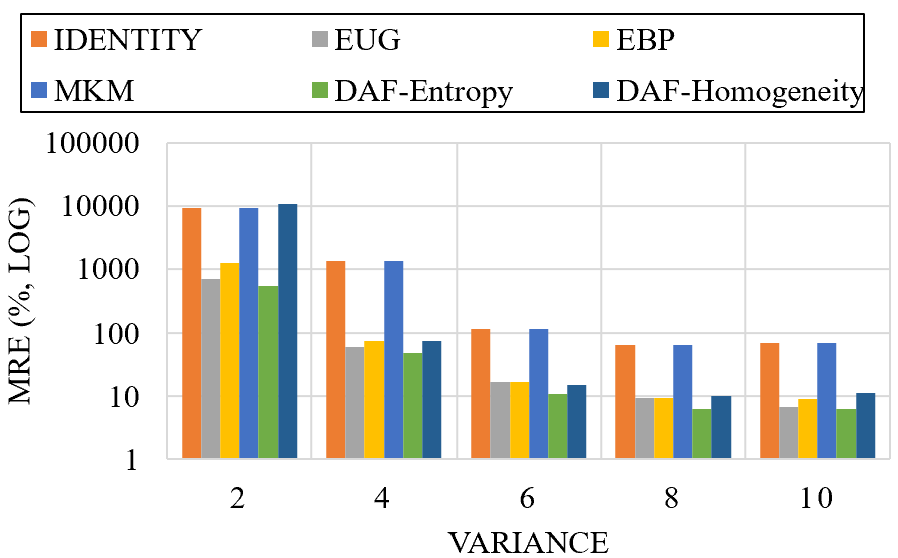}
	}
	\hfill
	\subfloat[4D, $\epsilon_{\text{tot}} = 0.3$.\label{s5}]{%
	\includegraphics[scale=.43]{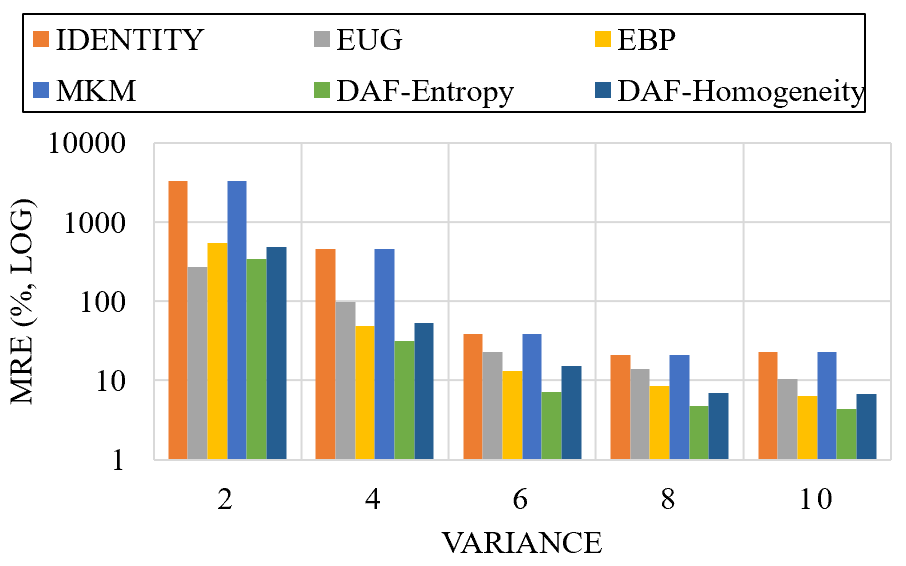}
	}
	\hfill
	\subfloat[4D, $\epsilon_{\text{tot}} = 0.5$.\label{s6}]{%
	\includegraphics[scale=.43]{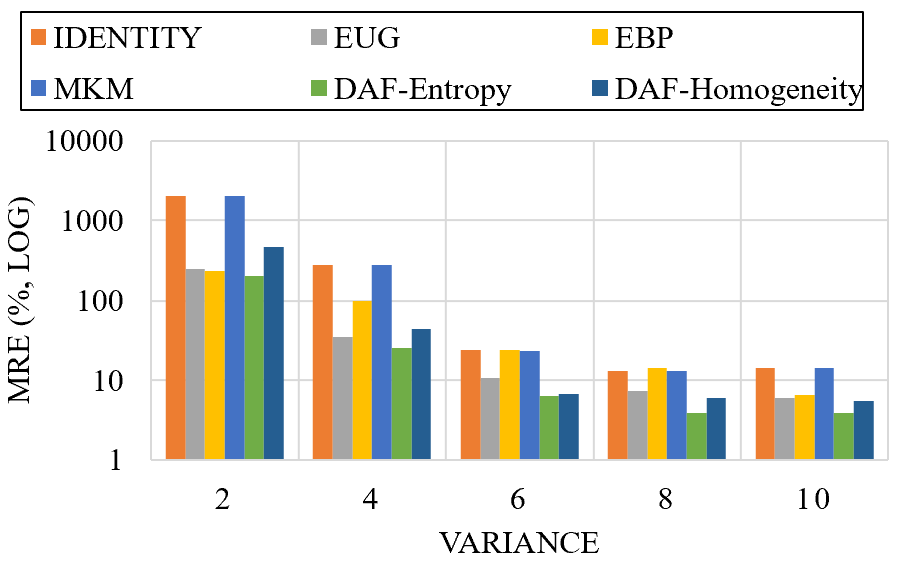}
	}
	\hfill
	\subfloat[6D, $\epsilon_{\text{tot}} = 0.1$.\label{s7}]{%
	\includegraphics[scale=.43]{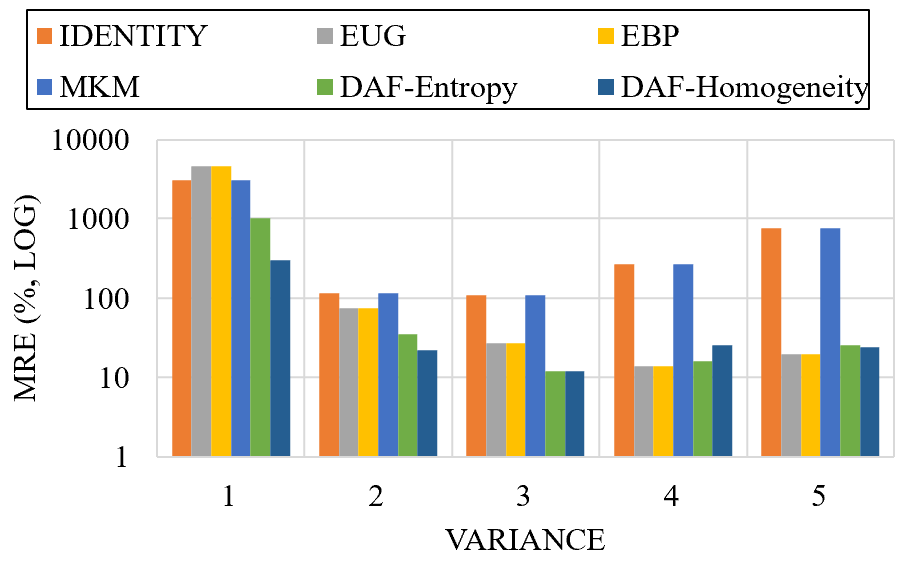}
	}
	\hfill
	\subfloat[6D, $\epsilon_{\text{tot}} = 0.3$.\label{s8}]{%
	\includegraphics[scale=.43]{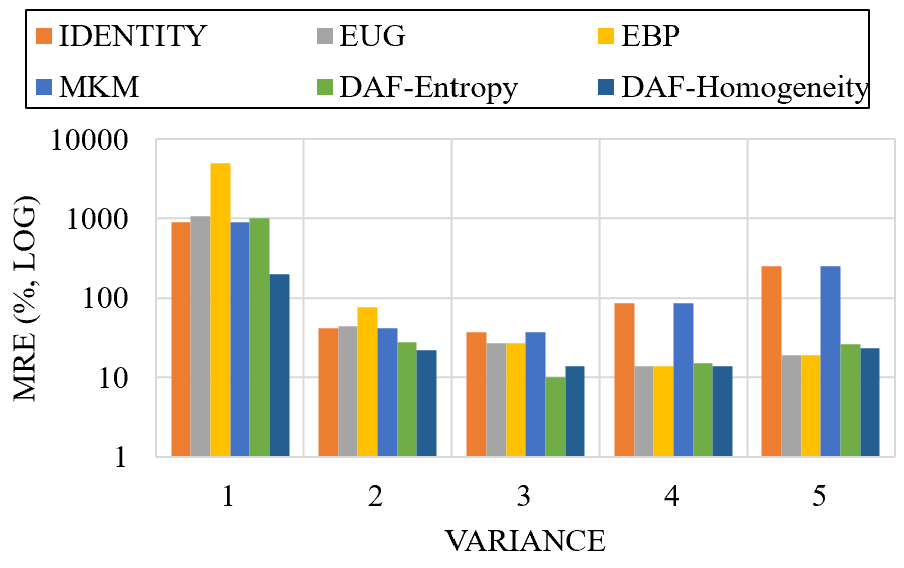}
	}
	\hfill
	\subfloat[6D, $\epsilon_{\text{tot}} = 0.5$.\label{s9}]{%
	\includegraphics[scale=.43]{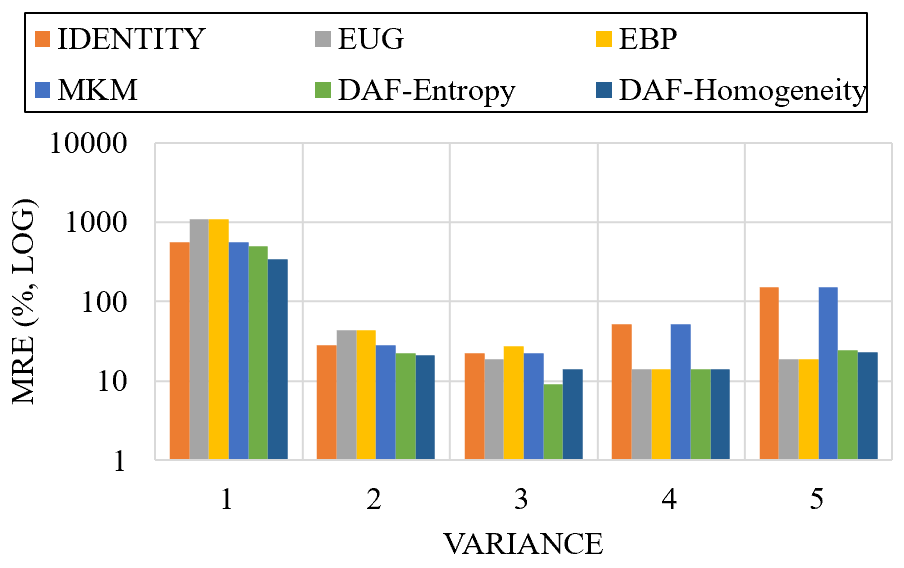}
	}
	\caption{Synthetic dataset results, Gaussian distribution, random shape and size queries.}
	\label{Fig: synthetic datasets}
\end{figure*}

Based on the real location data, we construct origin-destination matrices: in each city, $300,000$ trajectories are sampled, and their origin, destination and intermediate points are included in the OD matrix. The data are stored as a multi-dimensional frequency matrix generated as follows: the map of each city is discretized to a $1000\times 1000$ grid, and for every trajectory with the origin coordinates of $(x_o,y_o)$ and destination coordinates of $(x_d,y_d)$, the element $F[x_o,y_o,x_d,y_d]$ in the frequency matrix is incremented by one. A similar process is conducted for intermediate points, with the distinction that the matrix dimension count increases. 

The evaluation metric used to compare the results is Mean Relative Error (MRE), formally defined in Section~\ref{Sec: system model}, Eq.~\eqref{Equation: MRE}. We evaluate the accuracy of considered approaches on the basis of:

\begin{itemize}
    \item {\em Varying Data Skewness/Distribution.} The generation of synthetic datasets is conducted for Gaussian and Zipfian random variables with distinct variances;  for real-world datasets we select cities with a wide range of skewness properties. 
    \item {\em Varying Query Shape/Size.} Each data point in our experiments is the average MRE result of $1000$ queries generated based on random shapes and sizes. Additionally, the impact of small, medium and large queries is evaluated.
    \item {\em Varying Privacy Budget.} The experiments consider three privacy budget values of $0.1,\, 0.3,\, 0.5$ modeling high, moderate, and low privacy  constraints.
    \item {\em Varying  dimensionality.} We run experiments on frequency matrices with dimensionality from two to six. 
\end{itemize}

{\bf Compared Approaches.}
Table~\ref{Table: contribution} provides a summary and corresponding references for each of the algorithms used in our evaluation. More details about each of the baselines are provided in Section~\ref{Sec: literature reivew}. In total, we consider six techniques: IDENTITY, EUG, EBP, MKM, DAF-Entropy, and DAF-Homogeneity. 

\begin{figure*}[t]
	\subfloat[2D.\label{z1}]{%
	\includegraphics[scale=.23]{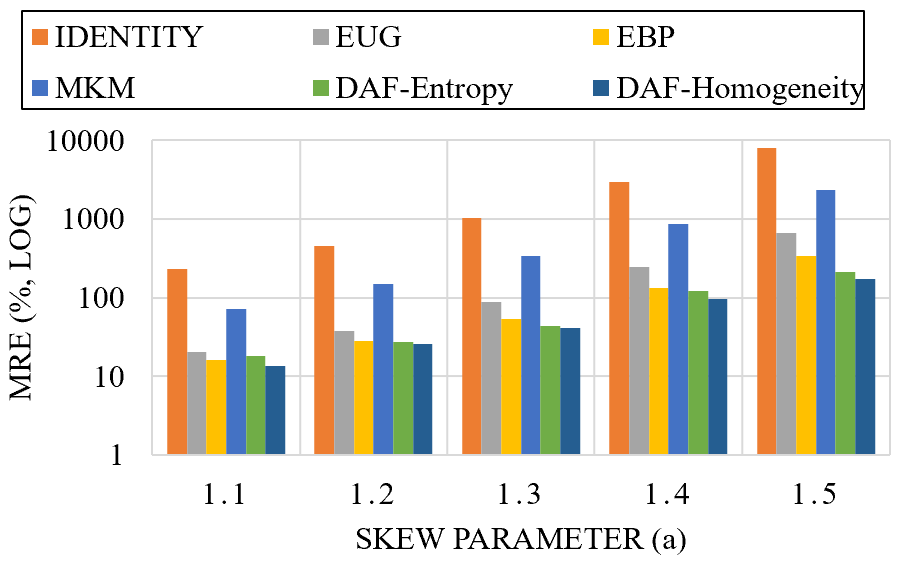}
	}
	\hfill
	\subfloat[4D. \label{z2}]{%
	\includegraphics[scale=.23]{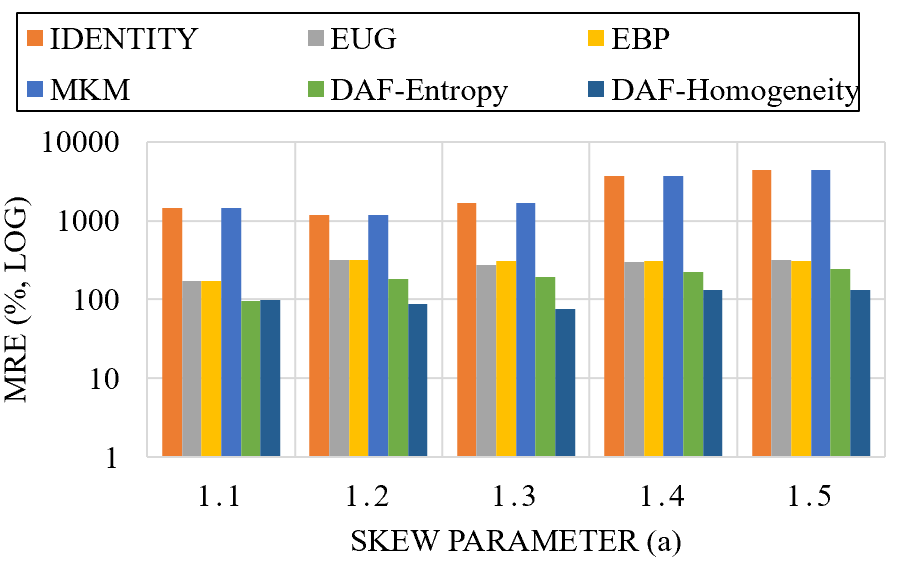}
	}
	\hfill
	\subfloat[6D. \label{z3}]{%
	\includegraphics[scale=.23]{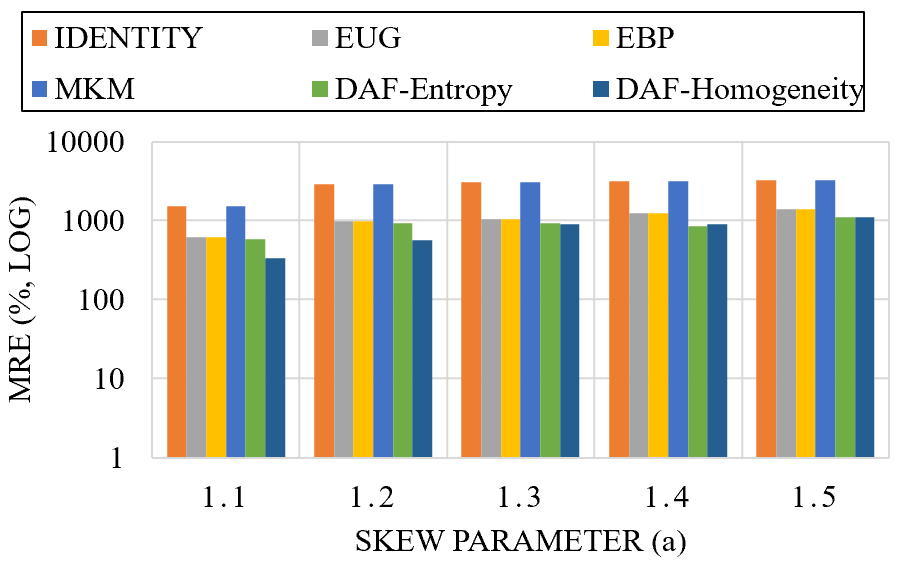}
	}
	\caption{Synthetic dataset results, Zipf distribution, random shape and size queries, $\epsilon_{\text{tot}} = 0.1$.}
	\label{Fig: synthetic datasets zipf distribution}
\end{figure*}	

\begin{figure*}[t]
	\hfill
	\subfloat[New York, random queries]{%
	\includegraphics[scale=.173]{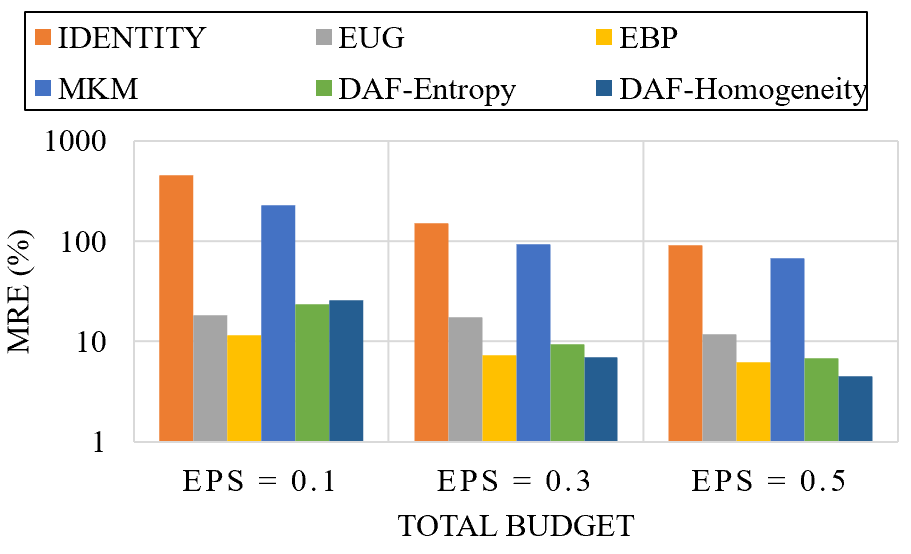}
	}
	\hfill
	\subfloat[New York, 1\% query coverage]{%
	\includegraphics[scale=.173]{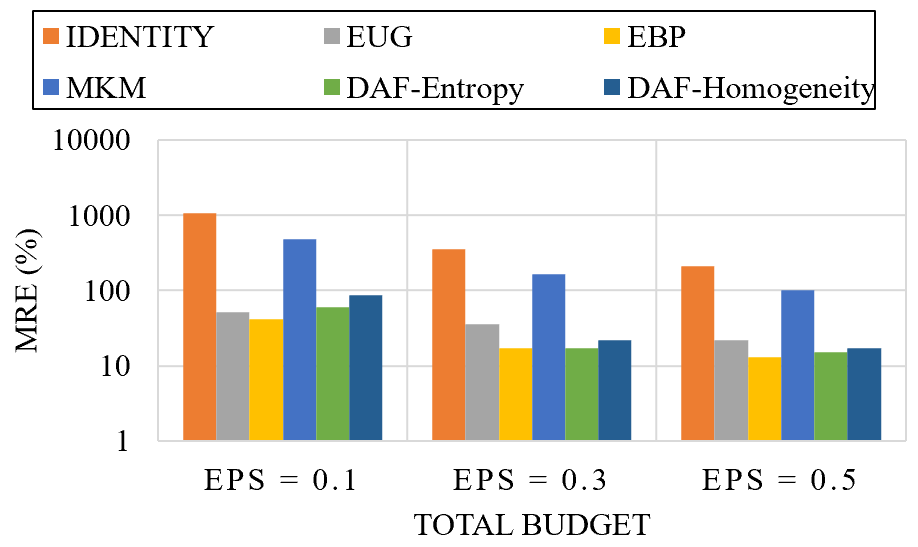}
	}
	\hfill
	\subfloat[New York, 5\% query coverage]{%
	\includegraphics[scale=.173]{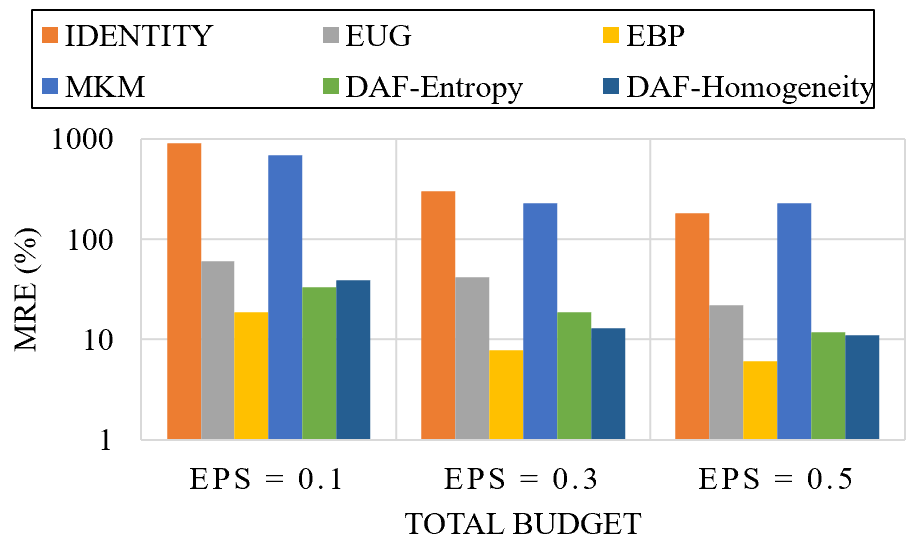}
	}
	\hfill
	\subfloat[New York, 10\% query coverage]{%
	\includegraphics[scale=.173]{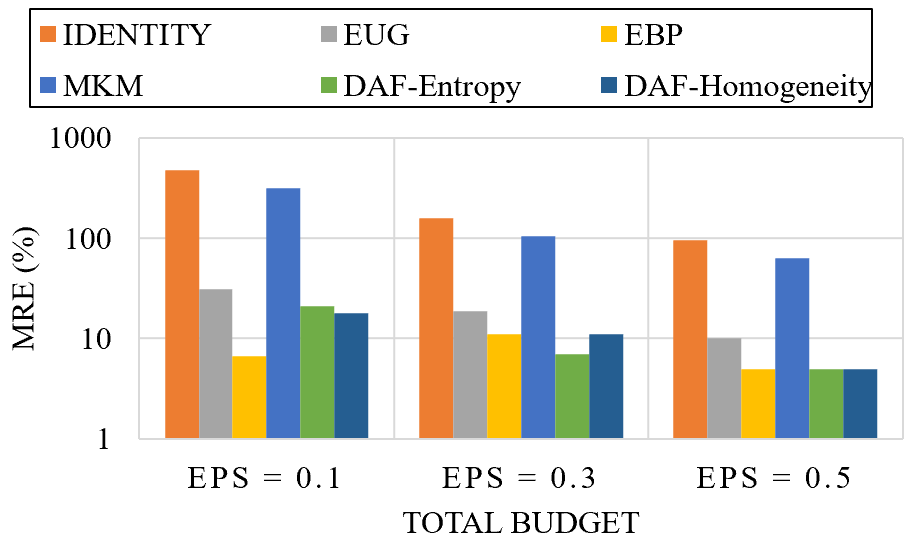}
	}
	\hfill
	\subfloat[Denver, random queries]{%
	\includegraphics[scale=.173]{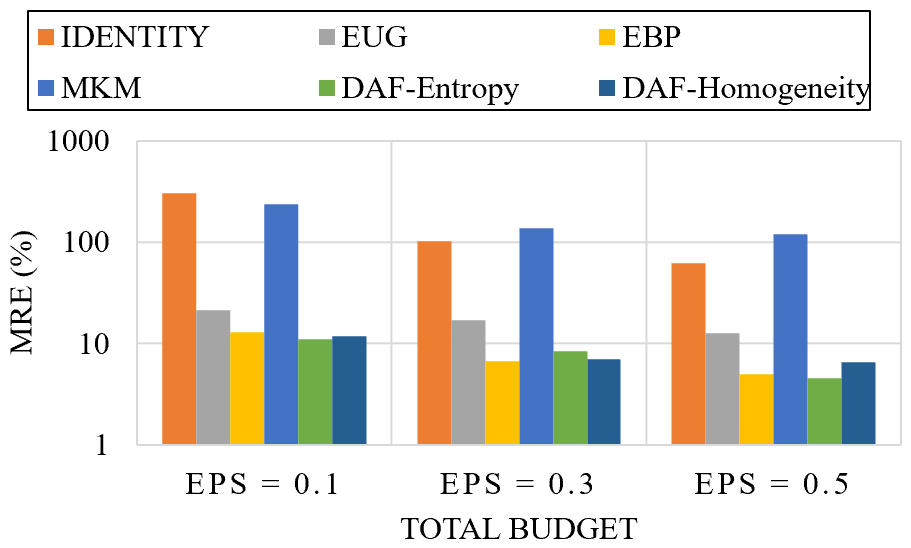}
	}
	\hfill
	\subfloat[Denver, 1\% query coverage]{%
	\includegraphics[scale=.173]{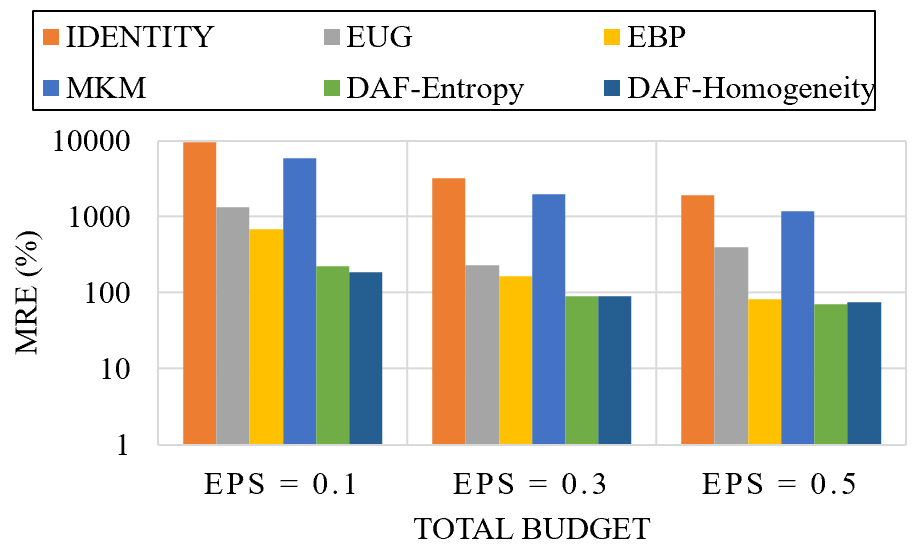}
	}
	\hfill
	\subfloat[Denver, 5\% query coverage]{%
	\includegraphics[scale=.173]{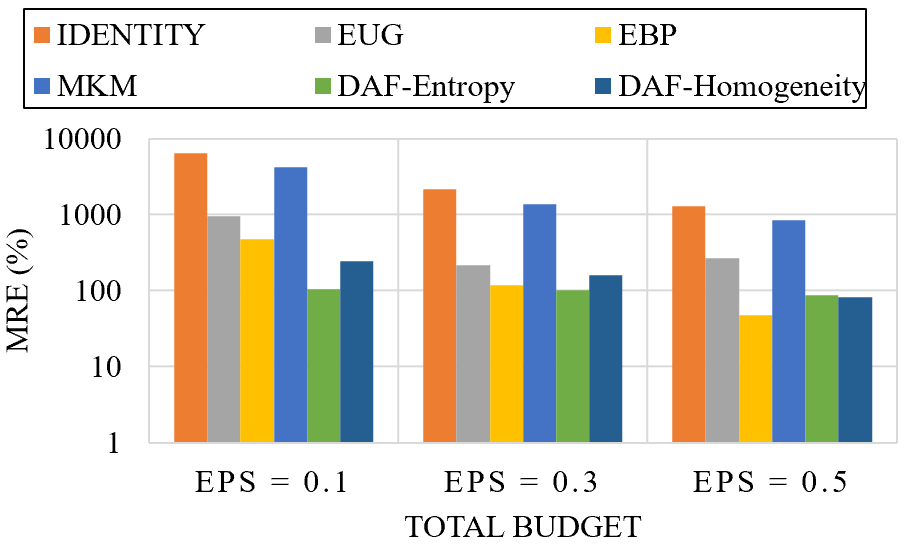}
	}
	\hfill
	\subfloat[Denver, 10\% query coverage]{%
	\includegraphics[scale=.173]{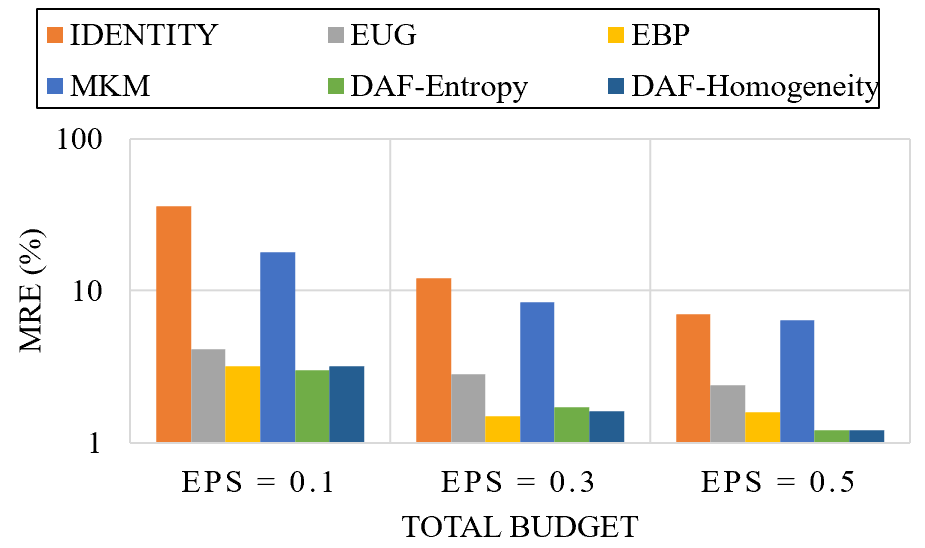}
	}
	\hfill
	\subfloat[Detroit, random queries]{%
	\includegraphics[scale=.173]{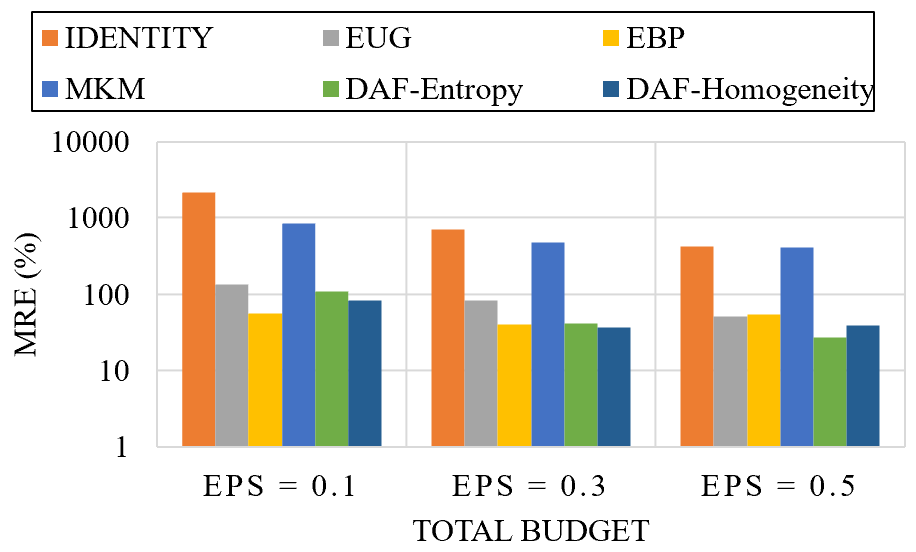}
	}
	\hfill
	\subfloat[Detroit, 1\% query coverage]{%
	\includegraphics[scale=.173]{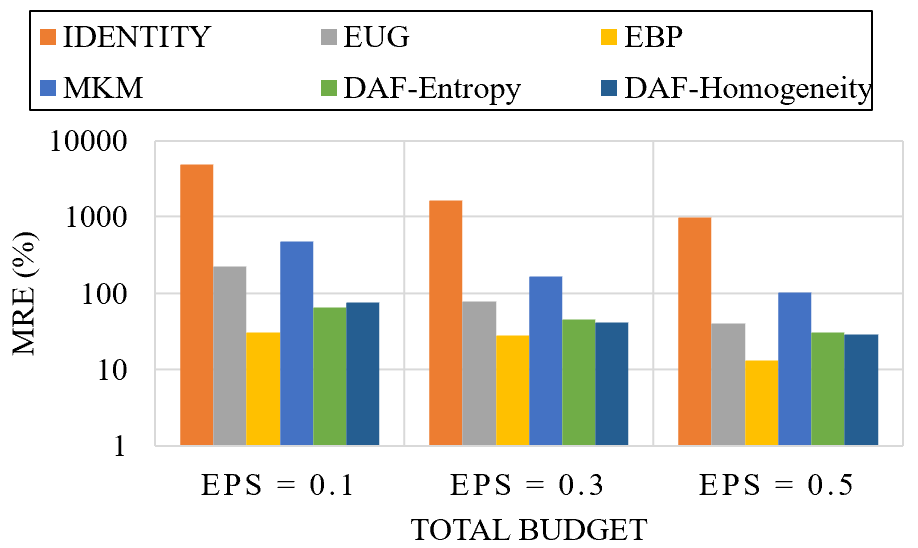}
	}
	\hfill
	\subfloat[Detroit, 5\% query coverage]{%
	\includegraphics[scale=.173]{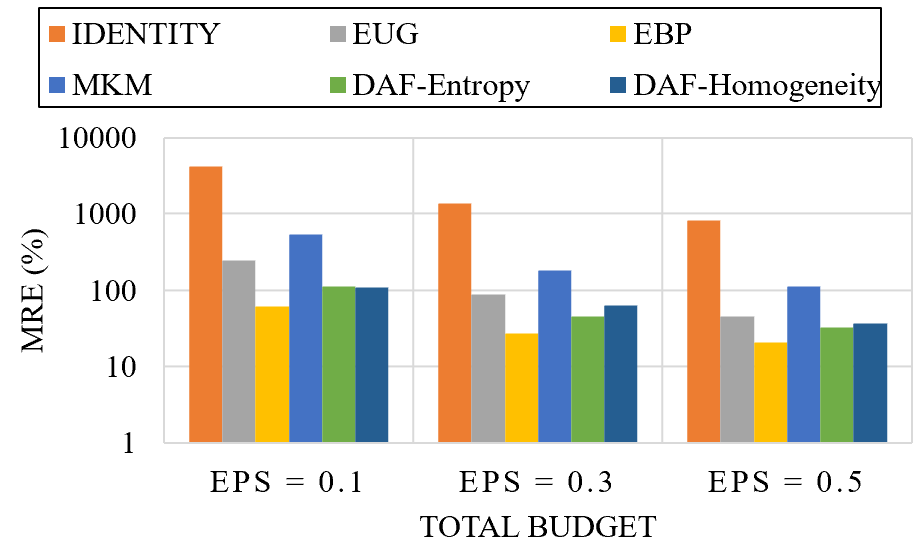}
	}
	\hfill
	\subfloat[Detroit, 10\% query coverage]{%
	\includegraphics[scale=.173]{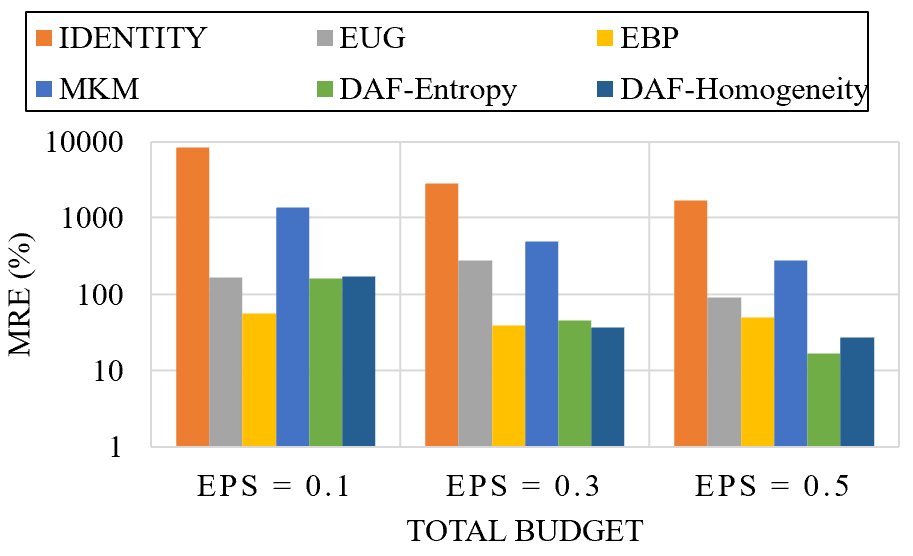}
	}
	\caption{Population histograms in 2D for real datasets.}
	\label{Fig: population histograms 2D all}
\end{figure*}

\subsection{Results on Synthetic Datasets}\label{Subsec: Results on Synthetic Datasets}

Figure~\ref{Fig: synthetic datasets} presents evaluation results on synthetic datasets. For each distinct dimensionality (i.e., row), we consider low, medium and high privacy budget settings. The width of frequency matrices in each dimension is set to $\sqrt[d]{N}$.

For the 2D case, results are shown in Figures~\ref{s1}-\ref{s3}. EBP and DAF-Entropy provide superior accuracy compared to other techniques, followed by DAF-Homogeneity and EUG. The MKM and IDENTITY algorithms exhibit similar performance, and we observed that MKM is reaching the maximum granularity for the frequency matrix. This is justified by the fact that the MKM approach does not follow the {\em epsilon-scale} exchangeability principle identified in~\cite{ashwin}. In general, there exist two scenarios in which data-independent algorithms perform better: (i) the data points are distributed almost uniformly, (i.e., high variance) and (ii) the data points are densely populated in the cluster center in a handful of matrix entries (i.e., low variance). 
The superior performance of the DAF framework becomes more evident in higher dimensions. In almost all experiments conducted, the DAF framework outperformed the data-independent sanitization approaches. Among the two objective functions that we developed for DAF, DAF-Entropy generally outperforms  DAF-Homogeneity. 

We also evaluate the studied approaches for Zipf synthetic distribution of data. Figure~\ref{Fig: synthetic datasets zipf distribution} shows similar relative trends, with the proposed approaches outperforming existing work by an order of magnitude. The error increases as the skew parameter $a$ increases.

\begin{figure*}[t]
	\subfloat[New York, random queries]{%
	\includegraphics[scale=.173]{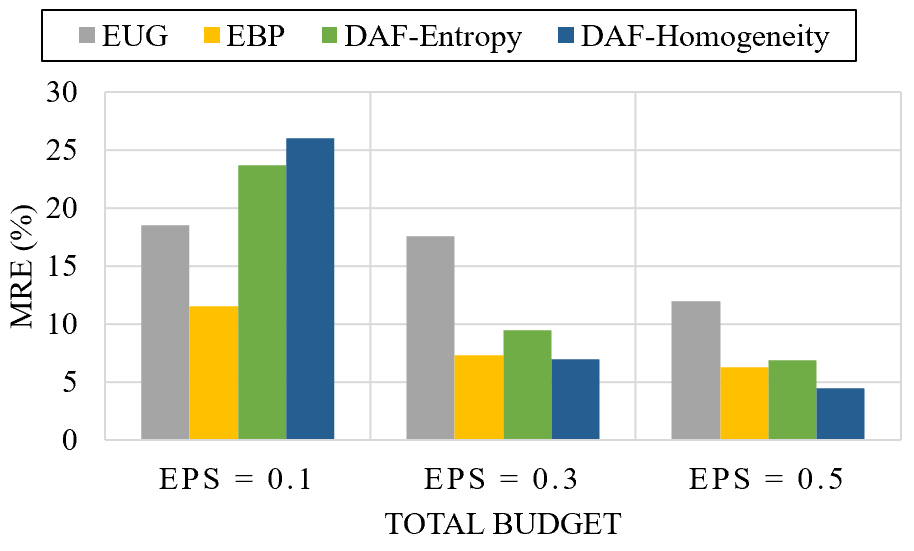}
	}
	\hfill
	\subfloat[New York, 1\% query coverage]{%
	\includegraphics[scale=.173]{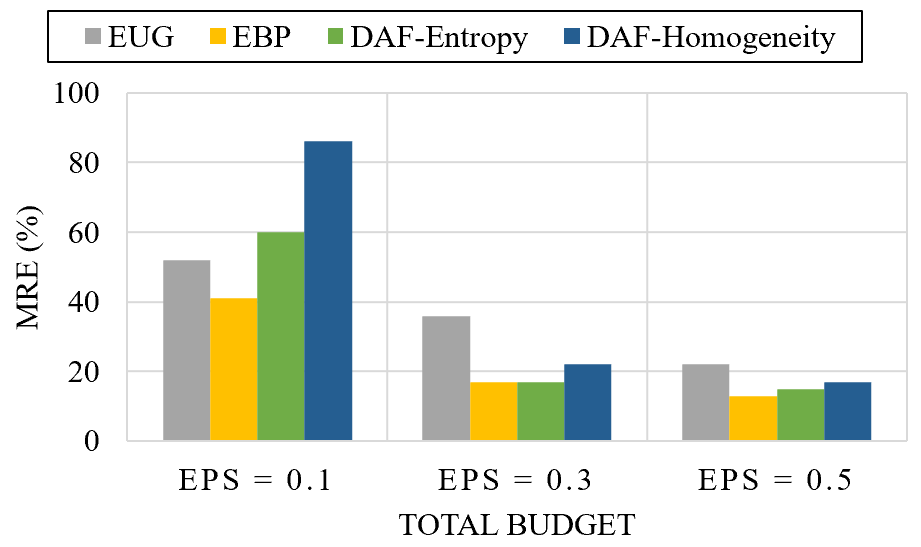}
	}
	\hfill
	\subfloat[New York, 5\% query coverage]{%
	\includegraphics[scale=.173]{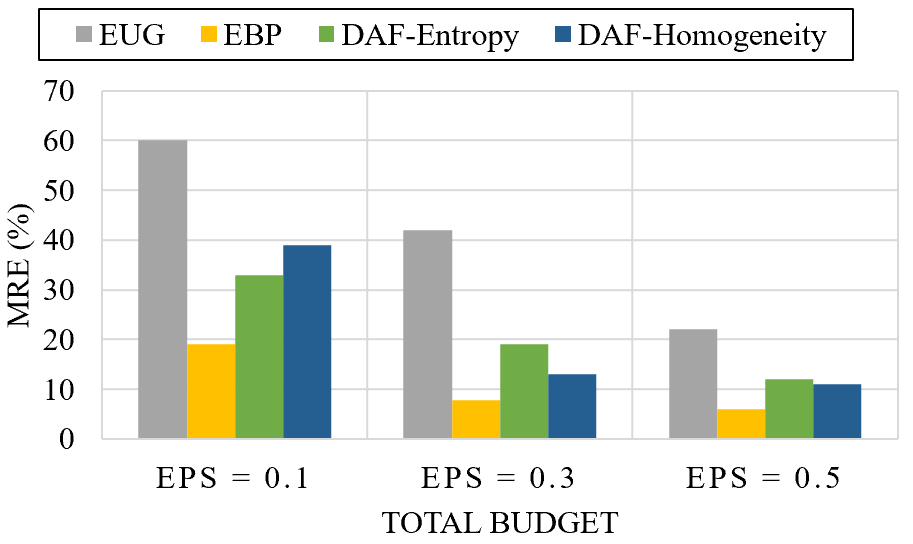}
	}
	\hfill
	\subfloat[New York, 10\% query coverage]{%
	\includegraphics[scale=.215]{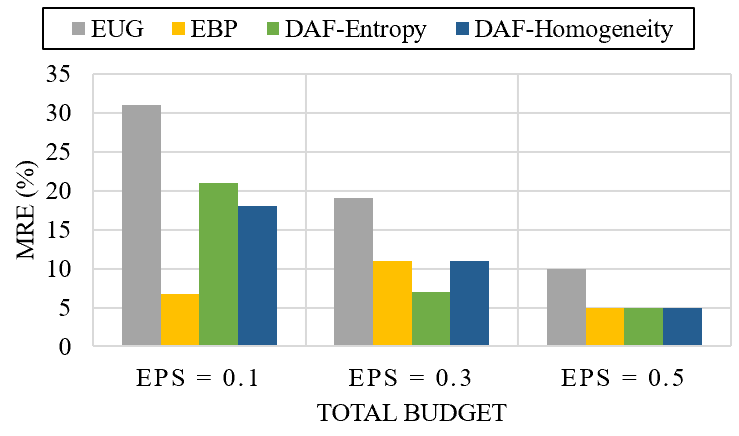}
	}
	\hfill
	\subfloat[Denver, random queries]{%
	\includegraphics[scale=.173]{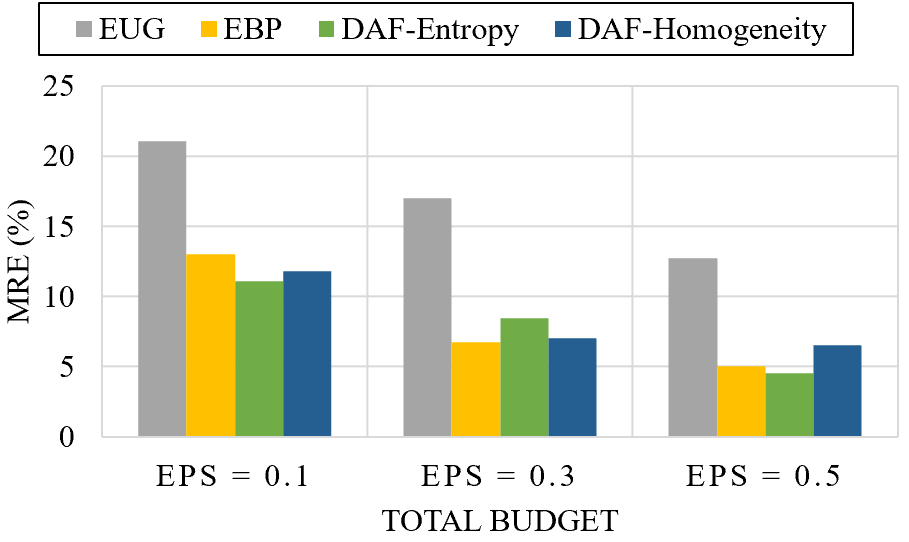}
	}
	\hfill
	\subfloat[Denver, 1\% query coverage]{%
	\includegraphics[scale=.173]{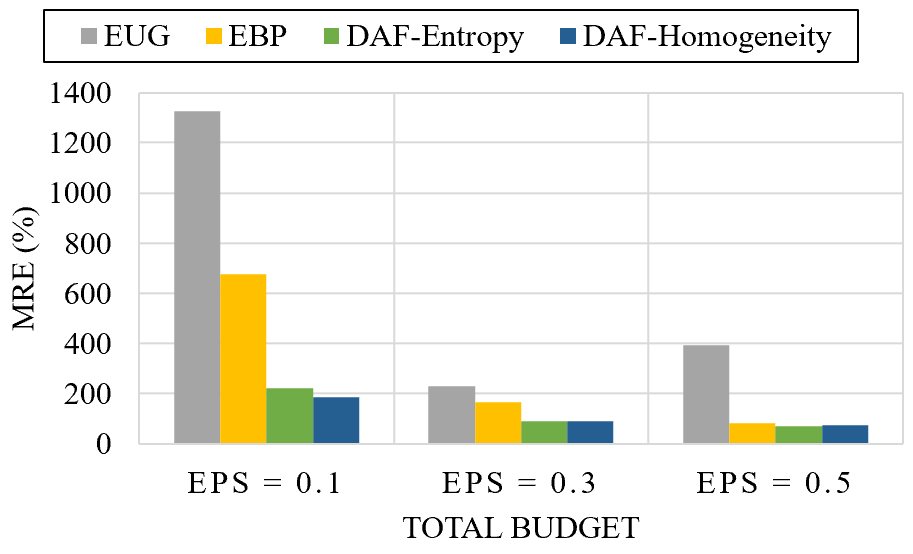}
	}
	\hfill
	\subfloat[Denver, 5\% query coverage]{%
	\includegraphics[scale=.173]{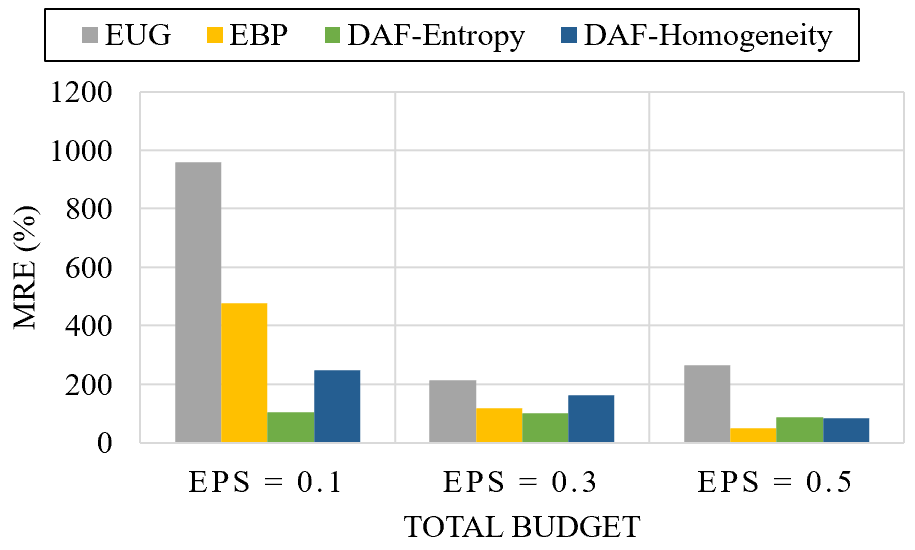}
	}
	\hfill
	\subfloat[Denver, 10\% query coverage]{%
	\includegraphics[scale=.173]{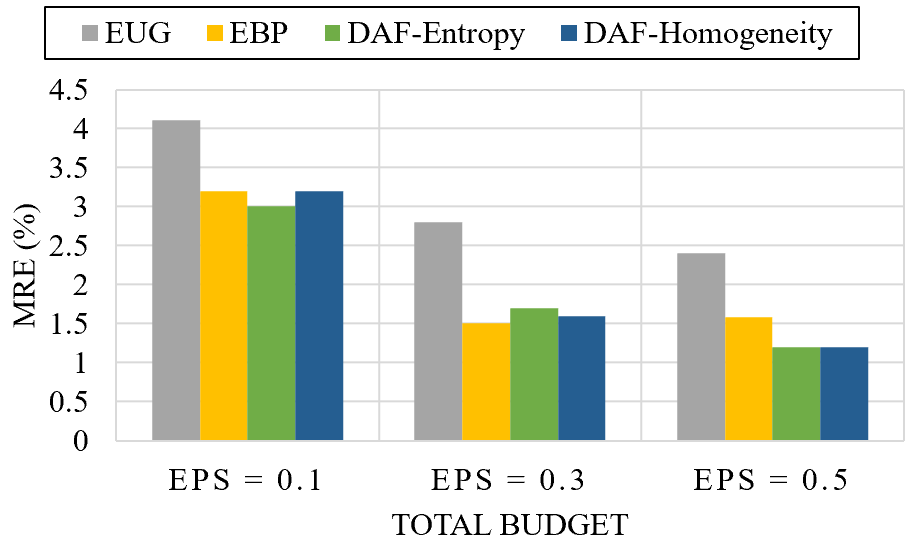}
	}
	\hfill
	\subfloat[Detroit, random queries]{%
	\includegraphics[scale=.173]{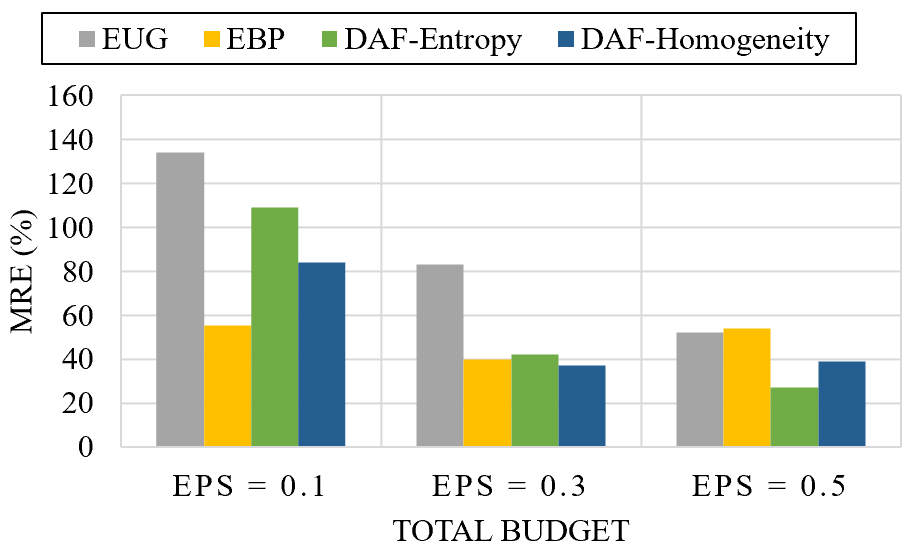}
	}
	\hfill
	\subfloat[Detroit, 1\% query coverage]{%
	\includegraphics[scale=.173]{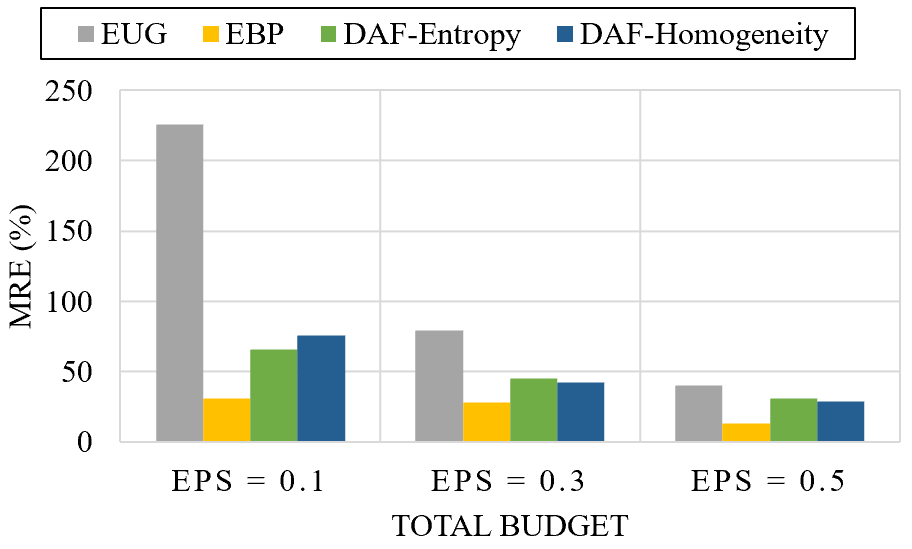}
	}
	\hfill
	\subfloat[Detroit, 5\% query coverage]{%
	\includegraphics[scale=.173]{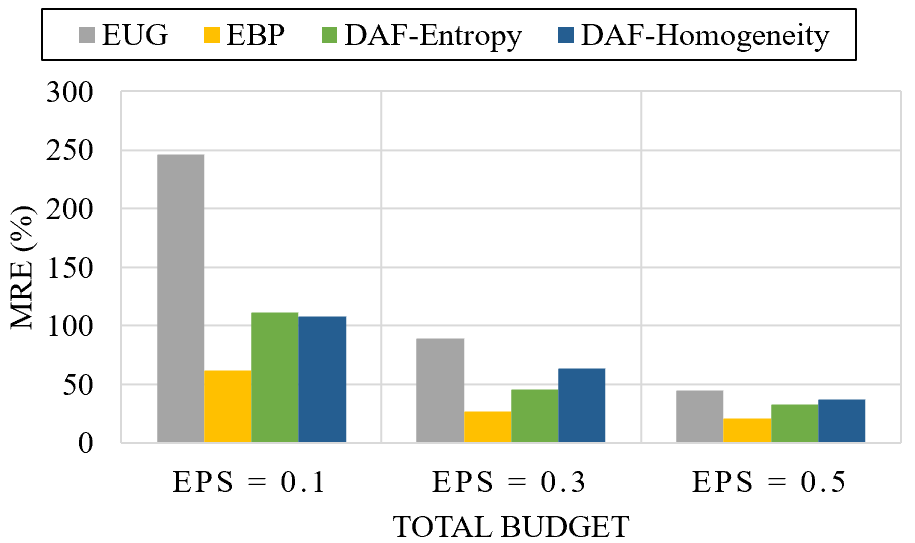}
	}
	\hfill
	\subfloat[Detroit, 10\% query coverage]{%
	\includegraphics[scale=.173]{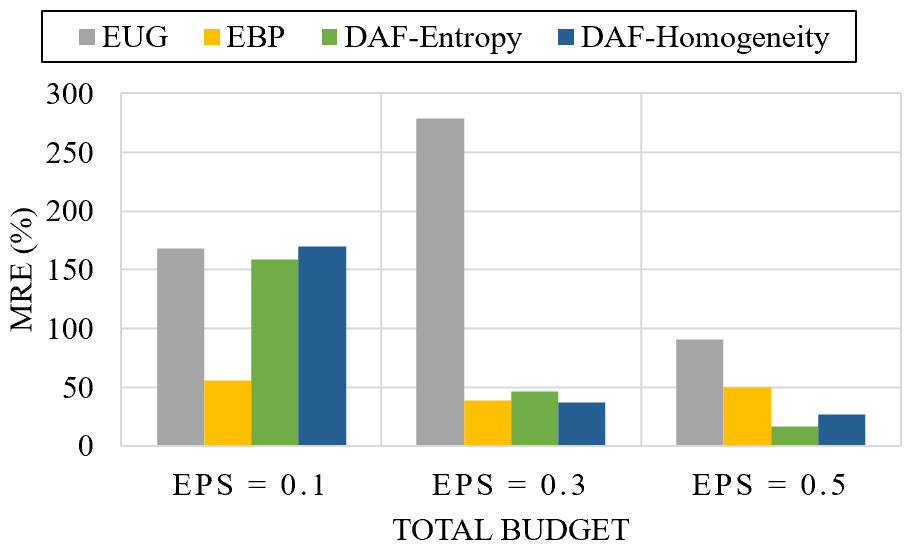}
	}
	\caption{Population histograms in 2D on real datasets, no baselines.}
	\label{Fig: population histograms 2D nobase}
\end{figure*}

\subsection{Results on Real-World Datasets}\label{Subsec: Results on Real-world Datasets}

Figure~\ref{Fig: population histograms 2D all} shows the accuracy of all studied methods on 2D data, for various query workloads: a mix of random queries, as well as fixed coverage queries with range from $1\%$ to $10\%$ of dataspace side. As in the case of synthetic data, the IDENTITY and MKM benchmarks underperform by an order of magnitude. For all methods, the error decreases when the query range increases, which is expected, since coarser queries can be accurately answered using most methods. However, the more challenging case is that of small query ranges, which provide more detailed information to the analyst. 

Due to their poor performance, we exclude IDENTITY and MKM from the rest of the results, and focus on studying the relative performance of the proposed approaches, illustrated in Figure~\ref{Fig: population histograms 2D nobase} on linear scale. The EUG algorithm results in poorer accuracy overall. For Detroit and New York EBP has performed better than competing techniques. The EBP and DAF results are comparable for the Denver datasets, with DAF-Homogeneity providing the highest accuracy. The EBP algorithm performs better in cities where the entropy of the population histogram is higher. This aligns with our expectations, as greater entropy can be an indicator of higher skewness, where EUG performs worse. When increasing the privacy budget, the error of all algorithms decreases consistently, since the noise required to satisfy the privacy bound becomes lower.  
Fig.~\ref{Fig: population histograms 4D} presents the results for higher-dimensionality matrices. Similar to the results observed for synthetic datasets, DAF-Entropy has superior accuracy on average compared to the other techniques. The relative accuracy gain achieved by DAF is observed to increase as the number of dimensions increases.

Table~\ref{tab:my-table} shows the execution time for all techniques. The DAF methods have faster execution time, because they adapt to data and do not perform unnecessary splits. In all cases, the proposed techniques complete execution in less than five minutes.

{\bf Discussion.}
Data-independent methods perform better when data are highly uniform or highly concentrated around the cluster center. However, most location datasets do not fall in either of these cases, hence there is need for carefully-designed density-aware approaches, like the ones we proposed.
In lower dimensions, the EBP algorithm outperforms competitor approaches on both real-world and synthetic datasets. In higher dimensions, the density-aware algorithms outperform data-independent algorithms. The improvement margin increases as the number of dimensions grows.  
On average, DAF-Entropy outperforms its homogeneity-based counterpart due to the additional budget required for evaluating homogeneity metrics of candidate splits in the latter.

\begin{figure*}[t]
	\subfloat[New York, random queries]{%
	\includegraphics[scale=.173]{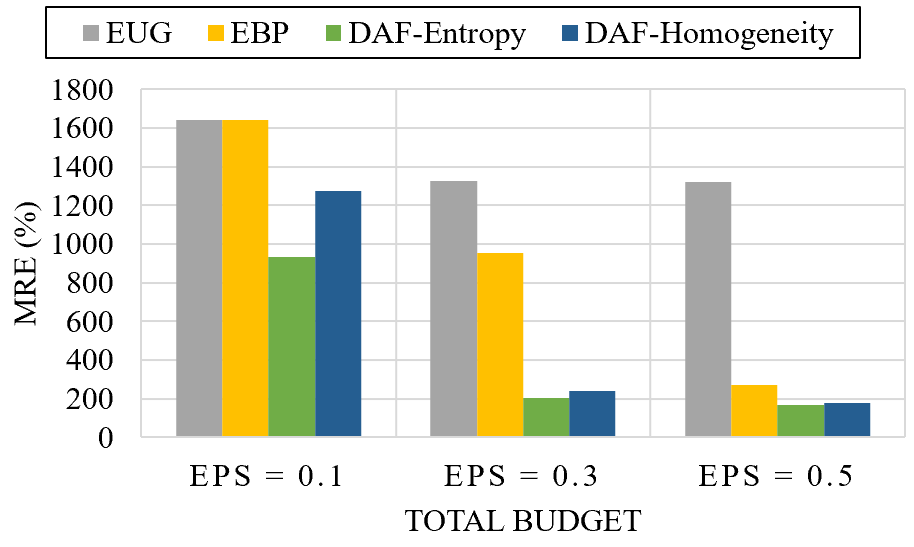}
	}
	\hfill
	\subfloat[New York, 1\% query coverage]{%
	\includegraphics[scale=.173]{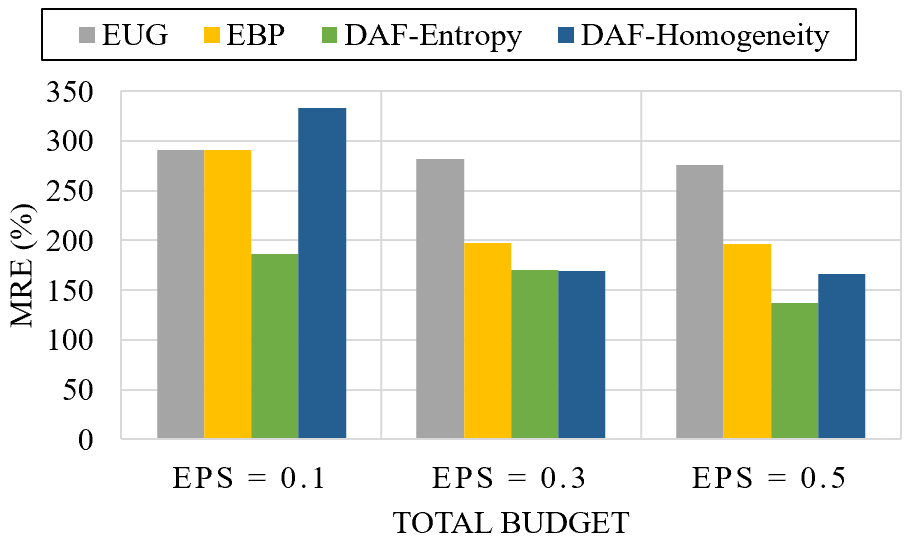}
	}
	\hfill
	\subfloat[New York, 5\% query coverage]{%
	\includegraphics[scale=.173]{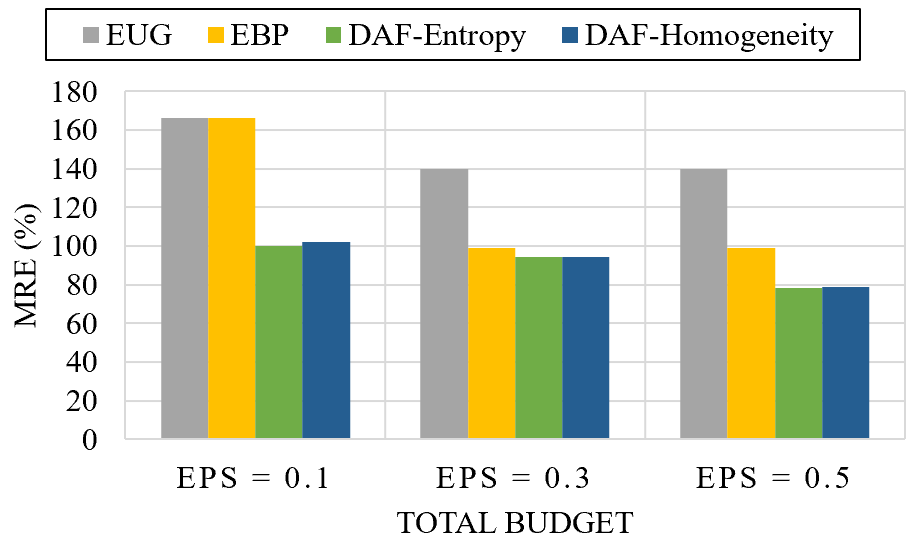}
	}
	\hfill
	\subfloat[New York, 10\% query coverage]{%
	\includegraphics[scale=.173]{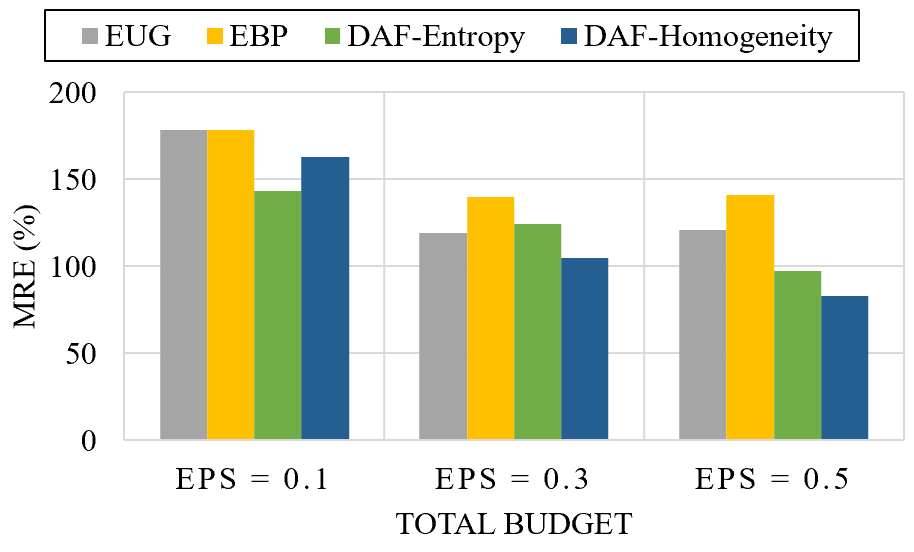}
	}	
	\hfill
	\subfloat[Denver, random queries]{%
	\includegraphics[scale=.173]{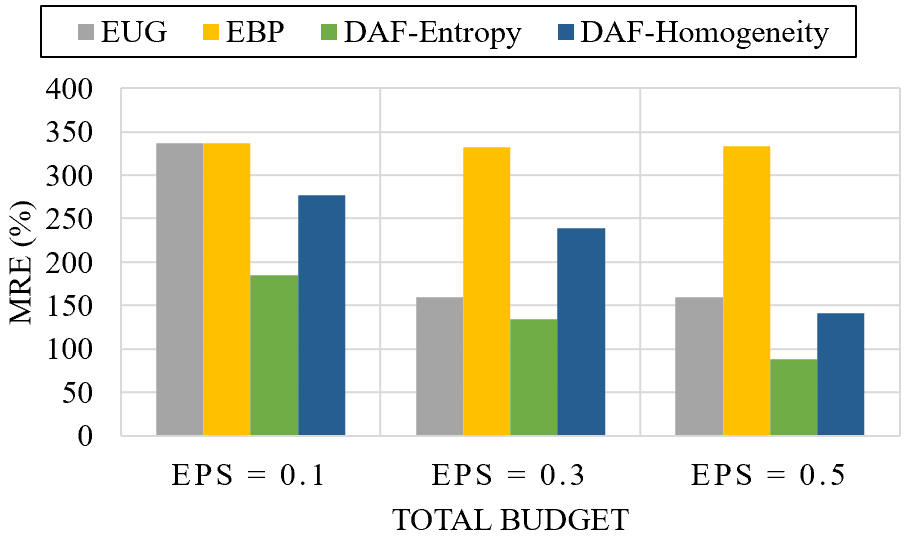}
	}
	\hfill
	\subfloat[Denver, 1\% query coverage]{%
	\includegraphics[scale=.173]{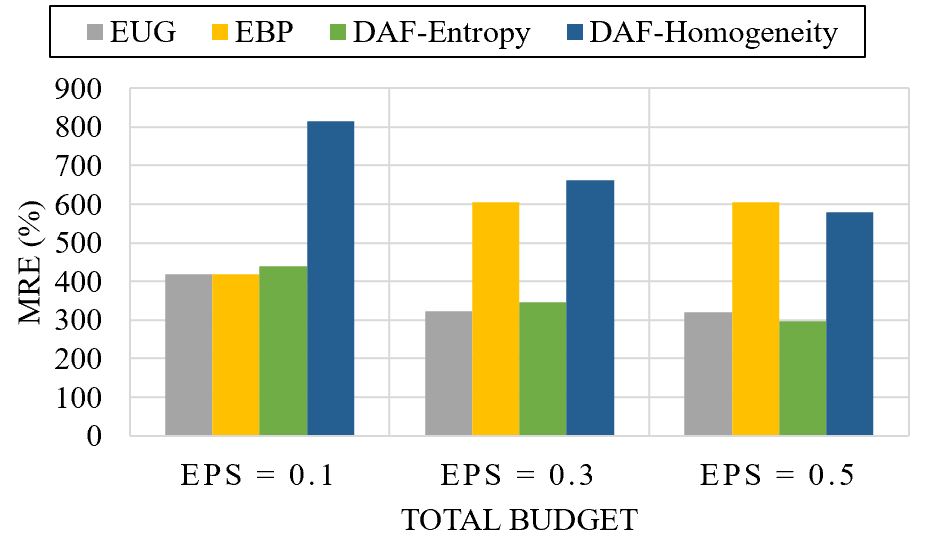}
	}
	\hfill
	\subfloat[Denver, 5\% query coverage]{%
	\includegraphics[scale=.173]{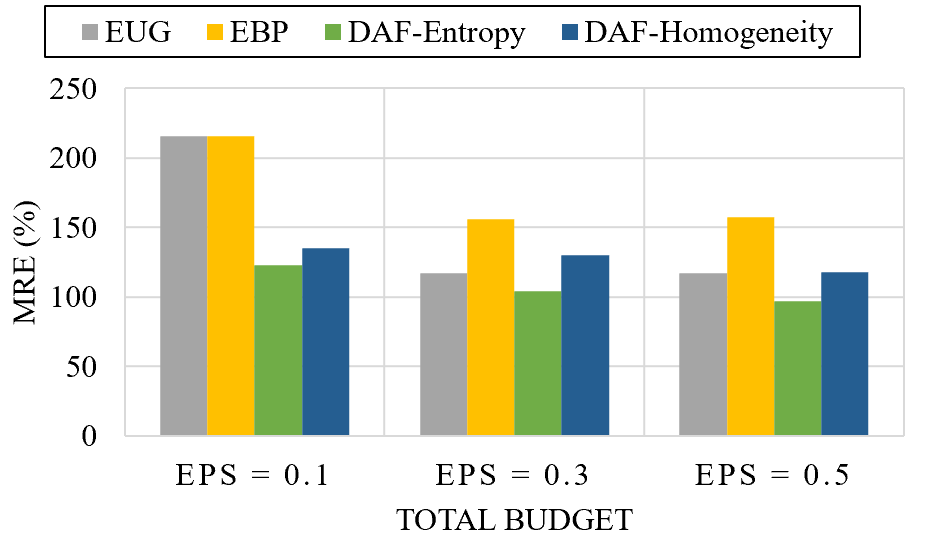}
	}
	\hfill
	\subfloat[Denver, 10\% query coverage]{%
	\includegraphics[scale=.173]{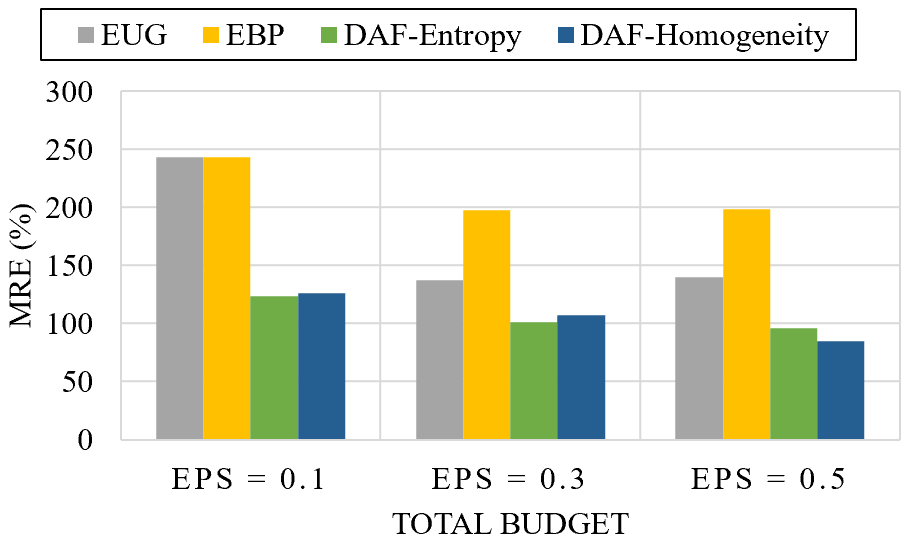}
	}
	\hfill
	\subfloat[Detroit, random queries]{%
	\includegraphics[scale=.173]{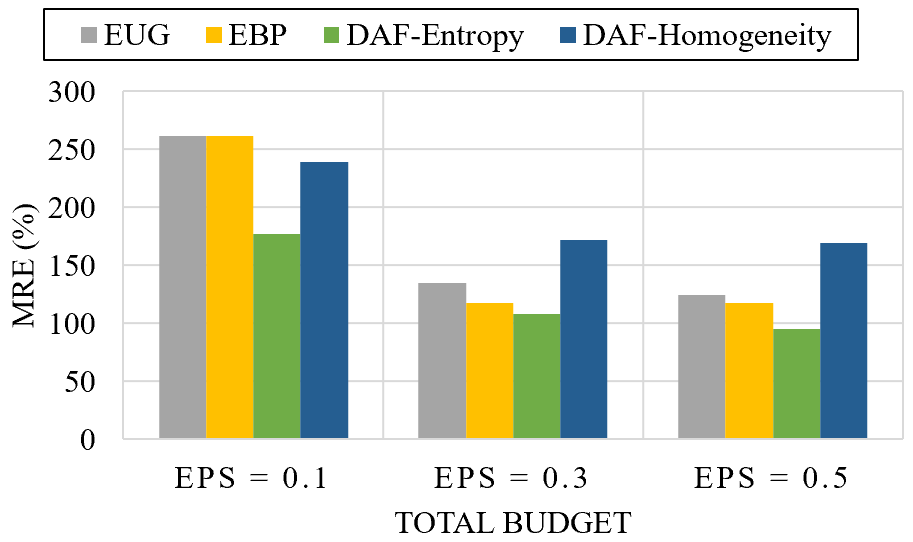}
	}
	\hfill
	\subfloat[Detroit, 1\% query coverage]{%
	\includegraphics[scale=.173]{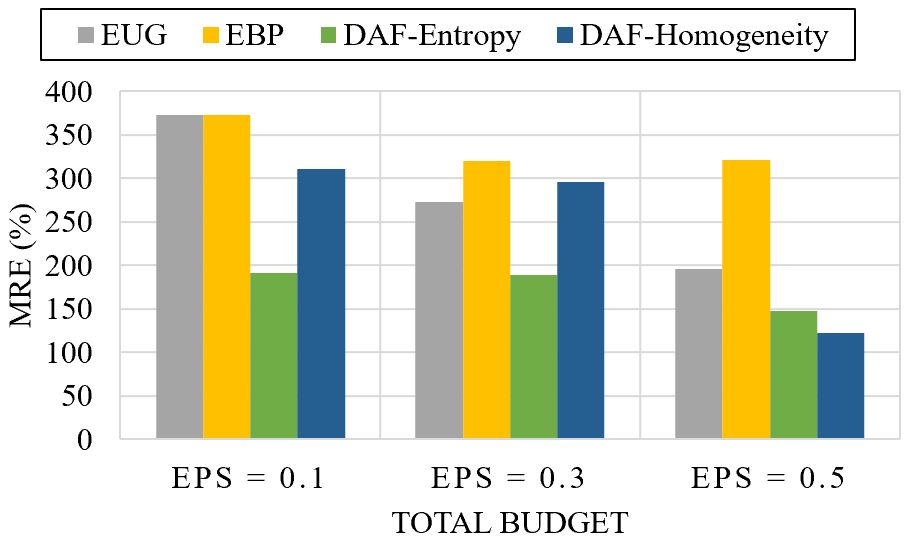}
	}
	\hfill
	\subfloat[Detroit, 5\% query coverage]{%
	\includegraphics[scale=.173]{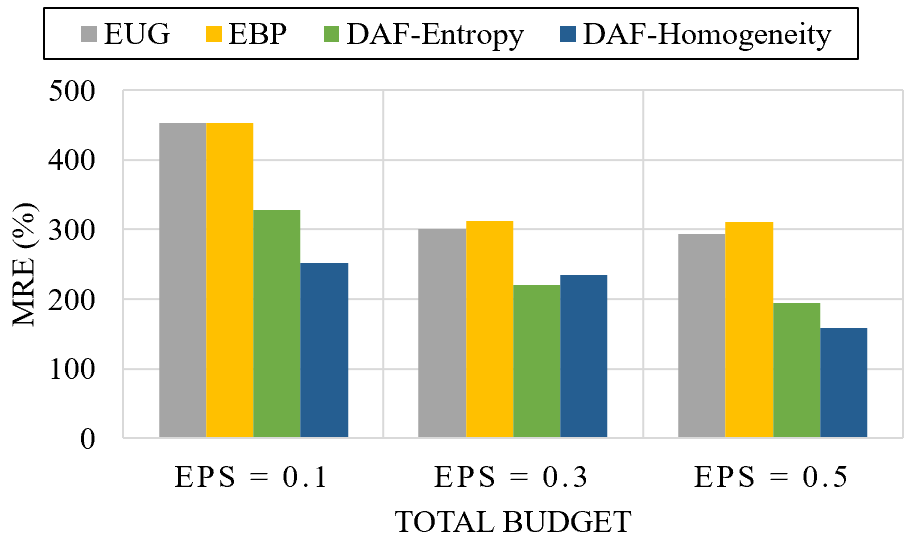}
	}
	\hfill
	\subfloat[Detroit, 10\% query coverage]{%
	\includegraphics[scale=.173]{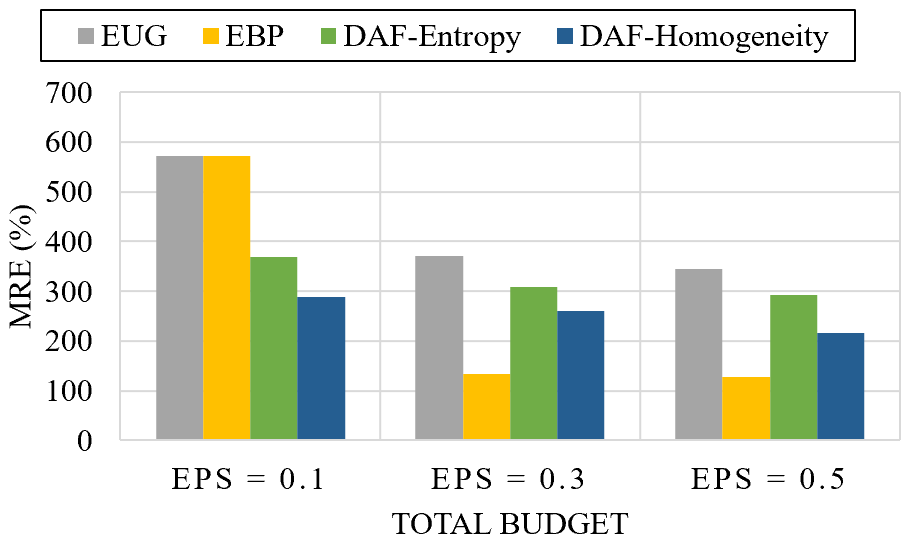}
	}
	\caption{Origin-Destination matrices in 4D, real datasets.}
	\label{Fig: population histograms 4D}
\end{figure*}

\begin{table}[]
\caption{Running time of algorithms (seconds), 2D, $\epsilon = 0.1$}
\label{tab:my-table}
\resizebox{.45\textwidth}{!}{%
\begin{tabular}{|l|l|l|l|l|l|l|}
\hline
         & IDENTITY & EUG   & EBP  & MKM  & DAF-Entropy & DAF-Homogeneity \\ \hline
New York & 89     & 87  & 87  & 177 & .47        & 0.5            \\ \hline
Denver   & 91     & 91 & 94 & 182   & 0.38        & 0.46           \\ \hline
Detroit  & 111     & 111  & 110 & 272   & 0.34        & 0.48            \\ \hline
\end{tabular}%
}
\end{table}

\section{Conclusion}\label{Sec: Conclusion}

We proposed a customized privacy-preserving approach for the publication of origin-destination matrices with intermediate stops in the context of differential privacy. Our proposed approaches provide the strong formal protection guarantees of differential privacy, while achieving superior accuracy to existing techniques that are designed for low-dimensionality location data and do not adapt well to data properties such as density variation. In future work, we plan to further improve accuracy by considering more sophisticated mechanisms in addition to Laplace noise addition. We will also investigate the correlation between location and semantic features of the geographical dataspace, which can provide additional accuracy in the case of semantic-centric queries (e.g., an analyst may be interested in trajectories that satisfy some semantic constraint, like {\em workplace-entertainment-sports sequences}, where the type of feature visited is more important than the actual geographical placement).

\newpage

\section*{Acknowledgement}
   This research has been funded in part by NSF grants 
   IIS-1910950, IIS-1909806, CNS-2027794,
   IIS-2128661 and CNS-2125530, 
   the USC Integrated Media Systems Center (IMSC), and an unrestricted cash gift from Microsoft Research. Any opinions, findings, and conclusions or recommendations expressed in this material are those of the author(s) and do not necessarily reflect the views of any of the sponsors such as the NSF.

\bibliographystyle{ACM-Reference-Format}
\bibliography{sample-base}


\begin{thebibliography}{21}


\ifx \showCODEN    \undefined \def \showCODEN     #1{\unskip}     \fi
\ifx \showDOI      \undefined \def \showDOI       #1{#1}\fi
\ifx \showISBNx    \undefined \def \showISBNx     #1{\unskip}     \fi
\ifx \showISBNxiii \undefined \def \showISBNxiii  #1{\unskip}     \fi
\ifx \showISSN     \undefined \def \showISSN      #1{\unskip}     \fi
\ifx \showLCCN     \undefined \def \showLCCN      #1{\unskip}     \fi
\ifx \shownote     \undefined \def \shownote      #1{#1}          \fi
\ifx \showarticletitle \undefined \def \showarticletitle #1{#1}   \fi
\ifx \showURL      \undefined \def \showURL       {\relax}        \fi
\providecommand\bibfield[2]{#2}
\providecommand\bibinfo[2]{#2}
\providecommand\natexlab[1]{#1}
\providecommand\showeprint[2][]{arXiv:#2}

\bibitem[\protect\citeauthoryear{Acs, Castelluccia, and Chen}{Acs
  et~al\mbox{.}}{2012}]%
        {acs2012differentially}
\bibfield{author}{\bibinfo{person}{Gergely Acs}, \bibinfo{person}{Claude
  Castelluccia}, {and} \bibinfo{person}{Rui Chen}.}
  \bibinfo{year}{2012}\natexlab{}.
\newblock \showarticletitle{Differentially private histogram publishing through
  lossy compression}. In \bibinfo{booktitle}{\emph{Intl. Conference on Data
  Mining}}. \bibinfo{pages}{1--10}.
\newblock


\bibitem[\protect\citeauthoryear{Chen, Acs, and Castelluccia}{Chen
  et~al\mbox{.}}{2012}]%
        {chen2012differentially}
\bibfield{author}{\bibinfo{person}{Rui Chen}, \bibinfo{person}{Gergely Acs},
  {and} \bibinfo{person}{Claude Castelluccia}.}
  \bibinfo{year}{2012}\natexlab{}.
\newblock \showarticletitle{Differentially private sequential data publication
  via variable-length n-grams}. In \bibinfo{booktitle}{\emph{ACM CCS}}.
  \bibinfo{pages}{638--649}.
\newblock


\bibitem[\protect\citeauthoryear{Cormode, Garofalakis, and Shekelyan}{Cormode
  et~al\mbox{.}}{2021}]%
        {cormode2021data}
\bibfield{author}{\bibinfo{person}{Graham Cormode}, \bibinfo{person}{Minos
  Garofalakis}, {and} \bibinfo{person}{Michael Shekelyan}.}
  \bibinfo{year}{2021}\natexlab{}.
\newblock \showarticletitle{Data-Independent Space Partitionings for
  Summaries}. In \bibinfo{booktitle}{\emph{Proceedings of the 40th ACM
  SIGMOD-SIGACT-SIGAI Symposium on Principles of Database Systems}}.
  \bibinfo{pages}{285--298}.
\newblock


\bibitem[\protect\citeauthoryear{Cormode, Procopiuc, Srivastava, Shen, and
  Yu}{Cormode et~al\mbox{.}}{2012b}]%
        {cormode2012differentially}
\bibfield{author}{\bibinfo{person}{Graham Cormode}, \bibinfo{person}{Cecilia
  Procopiuc}, \bibinfo{person}{Divesh Srivastava}, \bibinfo{person}{Entong
  Shen}, {and} \bibinfo{person}{Ting Yu}.} \bibinfo{year}{2012}\natexlab{b}.
\newblock \showarticletitle{Differentially private spatial decompositions}. In
  \bibinfo{booktitle}{\emph{IEEE Conf. on Data Engineering}}.
  \bibinfo{pages}{20--31}.
\newblock


\bibitem[\protect\citeauthoryear{Cormode, Procopiuc, Srivastava, and
  Tran}{Cormode et~al\mbox{.}}{2012a}]%
        {cormode2012differentially2}
\bibfield{author}{\bibinfo{person}{Graham Cormode}, \bibinfo{person}{Cecilia
  Procopiuc}, \bibinfo{person}{Divesh Srivastava}, {and}
  \bibinfo{person}{Thanh~TL Tran}.} \bibinfo{year}{2012}\natexlab{a}.
\newblock \showarticletitle{Differentially private summaries for sparse data}.
  In \bibinfo{booktitle}{\emph{Proceedings of the 15th International Conference
  on Database Theory}}. \bibinfo{pages}{299--311}.
\newblock


\bibitem[\protect\citeauthoryear{Diakonikolas, Li, and Schmidt}{Diakonikolas
  et~al\mbox{.}}{2018}]%
        {diakonikolas2018fast}
\bibfield{author}{\bibinfo{person}{Ilias Diakonikolas}, \bibinfo{person}{Jerry
  Li}, {and} \bibinfo{person}{Ludwig Schmidt}.}
  \bibinfo{year}{2018}\natexlab{}.
\newblock \showarticletitle{Fast and sample near-optimal algorithms for
  learning multidimensional histograms}. In
  \bibinfo{booktitle}{\emph{Conference On Learning Theory}}. PMLR,
  \bibinfo{pages}{819--842}.
\newblock


\bibitem[\protect\citeauthoryear{Dwork, McSherry, Nissim, and Smith}{Dwork
  et~al\mbox{.}}{2006}]%
        {dwork2006calibrating}
\bibfield{author}{\bibinfo{person}{Cynthia Dwork}, \bibinfo{person}{Frank
  McSherry}, \bibinfo{person}{Kobbi Nissim}, {and} \bibinfo{person}{Adam
  Smith}.} \bibinfo{year}{2006}\natexlab{}.
\newblock \showarticletitle{Calibrating noise to sensitivity in private data
  analysis}. In \bibinfo{booktitle}{\emph{Theory of cryptography}}.
  \bibinfo{pages}{265--284}.
\newblock


\bibitem[\protect\citeauthoryear{Hay, Machanavajjhala, Miklau, Chen, and
  Zhang}{Hay et~al\mbox{.}}{2016}]%
        {ashwin}
\bibfield{author}{\bibinfo{person}{Michael Hay}, \bibinfo{person}{Ashwin
  Machanavajjhala}, \bibinfo{person}{Gerome Miklau}, \bibinfo{person}{Yan
  Chen}, {and} \bibinfo{person}{Dan Zhang}.} \bibinfo{year}{2016}\natexlab{}.
\newblock \showarticletitle{Principled evaluation of differentially private
  algorithms using dpbench}. In \bibinfo{booktitle}{\emph{Proceedings of the
  2016 International Conference on Management of Data}}.
  \bibinfo{pages}{139--154}.
\newblock


\bibitem[\protect\citeauthoryear{Kernert, K{\"o}hler, and Lehner}{Kernert
  et~al\mbox{.}}{2015}]%
        {kernert2015spmacho}
\bibfield{author}{\bibinfo{person}{David Kernert}, \bibinfo{person}{Frank
  K{\"o}hler}, {and} \bibinfo{person}{Wolfgang Lehner}.}
  \bibinfo{year}{2015}\natexlab{}.
\newblock \showarticletitle{SpMacho-Optimizing Sparse Linear Algebra
  Expressions with Probabilistic Density Estimation.}. In
  \bibinfo{booktitle}{\emph{EDBT}}. \bibinfo{pages}{289--300}.
\newblock


\bibitem[\protect\citeauthoryear{Kernert, Lehner, and K{\"o}hler}{Kernert
  et~al\mbox{.}}{2016}]%
        {kernert2016topology}
\bibfield{author}{\bibinfo{person}{David Kernert}, \bibinfo{person}{Wolfgang
  Lehner}, {and} \bibinfo{person}{Frank K{\"o}hler}.}
  \bibinfo{year}{2016}\natexlab{}.
\newblock \showarticletitle{Topology-aware optimization of big sparse matrices
  and matrix multiplications on main-memory systems}. In
  \bibinfo{booktitle}{\emph{2016 IEEE 32nd International Conference on Data
  Engineering (ICDE)}}. IEEE, \bibinfo{pages}{823--834}.
\newblock


\bibitem[\protect\citeauthoryear{Lei}{Lei}{2011}]%
        {lei2011differentially}
\bibfield{author}{\bibinfo{person}{Jing Lei}.} \bibinfo{year}{2011}\natexlab{}.
\newblock \showarticletitle{Differentially private m-estimators}.
\newblock \bibinfo{journal}{\emph{Advances in Neural Information Processing
  Systems}}  \bibinfo{volume}{24} (\bibinfo{year}{2011}),
  \bibinfo{pages}{361--369}.
\newblock


\bibitem[\protect\citeauthoryear{Li, Hay, Miklau, and Wang}{Li
  et~al\mbox{.}}{2014}]%
        {li2014data}
\bibfield{author}{\bibinfo{person}{Chao Li}, \bibinfo{person}{Michael Hay},
  \bibinfo{person}{Gerome Miklau}, {and} \bibinfo{person}{Yue Wang}.}
  \bibinfo{year}{2014}\natexlab{}.
\newblock \showarticletitle{A data-and workload-aware algorithm for range
  queries under differential privacy}.
\newblock \bibinfo{journal}{\emph{Proceedings of the VLDB Endowment}}
  \bibinfo{volume}{7}, \bibinfo{number}{5} (\bibinfo{year}{2014}),
  \bibinfo{pages}{341--352}.
\newblock


\bibitem[\protect\citeauthoryear{Li and Miklau}{Li and Miklau}{2012}]%
        {li2012adaptive}
\bibfield{author}{\bibinfo{person}{Chao Li} {and} \bibinfo{person}{Gerome
  Miklau}.} \bibinfo{year}{2012}\natexlab{}.
\newblock \showarticletitle{An Adaptive Mechanism for Accurate Query Answering
  under Differential Privacy}.
\newblock \bibinfo{journal}{\emph{Proc. VLDB Endow.}} \bibinfo{volume}{5},
  \bibinfo{number}{6} (\bibinfo{year}{2012}), \bibinfo{pages}{514–525}.
\newblock


\bibitem[\protect\citeauthoryear{Li and Miklau}{Li and Miklau}{2013}]%
        {li2013optimal}
\bibfield{author}{\bibinfo{person}{Chao Li} {and} \bibinfo{person}{Gerome
  Miklau}.} \bibinfo{year}{2013}\natexlab{}.
\newblock \showarticletitle{Optimal error of query sets under the
  differentially-private matrix mechanism}. In \bibinfo{booktitle}{\emph{Intl.
  Conference on Database Theory}}. \bibinfo{pages}{272--283}.
\newblock


\bibitem[\protect\citeauthoryear{Qardaji, Yang, and Li}{Qardaji
  et~al\mbox{.}}{2013}]%
        {AG}
\bibfield{author}{\bibinfo{person}{Wahbeh Qardaji}, \bibinfo{person}{Weining
  Yang}, {and} \bibinfo{person}{Ninghui Li}.} \bibinfo{year}{2013}\natexlab{}.
\newblock \showarticletitle{Differentially private grids for geospatial data}.
  In \bibinfo{booktitle}{\emph{IEEE International conference on data
  engineering}}. IEEE, \bibinfo{pages}{757--768}.
\newblock


\bibitem[\protect\citeauthoryear{Shaham, Ghinita, Ahuja, Krumm, and
  Shahabi}{Shaham et~al\mbox{.}}{2021}]%
        {shaham2021htf}
\bibfield{author}{\bibinfo{person}{Sina Shaham}, \bibinfo{person}{Gabriel
  Ghinita}, \bibinfo{person}{Ritesh Ahuja}, \bibinfo{person}{John Krumm}, {and}
  \bibinfo{person}{Cyrus Shahabi}.} \bibinfo{year}{2021}\natexlab{}.
\newblock \showarticletitle{HTF: Homogeneous Tree Framework for
  Differentially-Private Release of Location Data}. In
  \bibinfo{booktitle}{\emph{Proceedings of the 29th International Conference on
  Advances in Geographic Information Systems}}. \bibinfo{pages}{184--194}.
\newblock


\bibitem[\protect\citeauthoryear{Veraset}{Veraset}{[n.d.]}]%
        {datarade}
\bibfield{author}{\bibinfo{person}{Veraset}.}
  \bibinfo{year}{[n.d.]}\natexlab{}.
\newblock \bibinfo{title}{{Veraset Movement data for the USA, The largest,
  deepest and broadest available movement dataset (anonymized GPS signals)}}.
\newblock
\newblock


\bibitem[\protect\citeauthoryear{Xiao, Wang, and Gehrke}{Xiao
  et~al\mbox{.}}{2010a}]%
        {xiao2010differential}
\bibfield{author}{\bibinfo{person}{Xiaokui Xiao}, \bibinfo{person}{Guozhang
  Wang}, {and} \bibinfo{person}{Johannes Gehrke}.}
  \bibinfo{year}{2010}\natexlab{a}.
\newblock \showarticletitle{Differential privacy via wavelet transforms}.
\newblock \bibinfo{journal}{\emph{IEEE Transactions on knowledge and data
  engineering}} \bibinfo{volume}{23}, \bibinfo{number}{8}
  (\bibinfo{year}{2010}), \bibinfo{pages}{1200--1214}.
\newblock


\bibitem[\protect\citeauthoryear{Xiao, Xiong, and Yuan}{Xiao
  et~al\mbox{.}}{2010b}]%
        {xiao2010differentially}
\bibfield{author}{\bibinfo{person}{Yonghui Xiao}, \bibinfo{person}{Li Xiong},
  {and} \bibinfo{person}{Chun Yuan}.} \bibinfo{year}{2010}\natexlab{b}.
\newblock \showarticletitle{Differentially private data release through
  multidimensional partitioning}. In \bibinfo{booktitle}{\emph{Workshop on
  Secure Data Management}}. Springer, \bibinfo{pages}{150--168}.
\newblock


\bibitem[\protect\citeauthoryear{Zhang, Xiao, and Xie}{Zhang
  et~al\mbox{.}}{2016}]%
        {zhang2016privtree}
\bibfield{author}{\bibinfo{person}{Jun Zhang}, \bibinfo{person}{Xiaokui Xiao},
  {and} \bibinfo{person}{Xing Xie}.} \bibinfo{year}{2016}\natexlab{}.
\newblock \showarticletitle{Privtree: A differentially private algorithm for
  hierarchical decompositions}. In \bibinfo{booktitle}{\emph{Proceedings of the
  2016 International Conference on Management of Data}}.
  \bibinfo{pages}{155--170}.
\newblock


\bibitem[\protect\citeauthoryear{Zhang, Chen, Xu, Meng, and Xie}{Zhang
  et~al\mbox{.}}{2014}]%
        {zhang2014towards}
\bibfield{author}{\bibinfo{person}{Xiaojian Zhang}, \bibinfo{person}{Rui Chen},
  \bibinfo{person}{Jianliang Xu}, \bibinfo{person}{Xiaofeng Meng}, {and}
  \bibinfo{person}{Yingtao Xie}.} \bibinfo{year}{2014}\natexlab{}.
\newblock \showarticletitle{Towards accurate histogram publication under
  differential privacy}. In \bibinfo{booktitle}{\emph{Proceedings of the 2014
  SIAM international conference on data mining}}. SIAM,
  \bibinfo{pages}{587--595}.
\newblock


\end{thebibliography}

%
% ...

\end{document}